\documentclass[a4paper,10pt,numbers=noenddot,version=first,twoside=on,headinclude=true,footinclude=false,headsepline=true]{scrartcl}
\usepackage[utf8]{inputenc}
\usepackage[T1]{fontenc}
\usepackage{color}
\usepackage{graphics}
\usepackage[ngerman,english]{babel}

\usepackage{amssymb}
\usepackage{amsmath}
\usepackage{amsfonts}
\usepackage{amsopn}

\usepackage{units}
\usepackage[all,cmtip]{xy}
\providecommand{\R}{\mathbb{R}}
\providecommand{\Q}{\mathbb{Q}}
\providecommand{\C}{\mathbb{C}}
\renewcommand{\C}{\mathbb{C}}

\providecommand{\T}{\mathbb{T}}
\providecommand{\N}{\mathbb{N}}
\providecommand{\Z}{\mathbb{Z}}
\providecommand{\eps}{\varepsilon}

\providecommand{\ran}{\mathrm{ran} \, }

\providecommand{\abs}[1]{\left \lvert #1 \right \rvert}
\providecommand{\sabs}[1]{\lvert #1 \vert}
\providecommand{\babs}[1]{\bigl \lvert #1 \bigr \rvert}
\providecommand{\Babs}[1]{\Bigl \lvert #1 \Bigr \rvert}
\providecommand{\norm}[1]{\left \lVert #1 \right \rVert}
\providecommand{\snorm}[1]{\lVert #1 \rVert}
\providecommand{\bnorm}[1]{\bigl \lVert #1 \bigr \rVert}

\providecommand{\scpro}[2]{\left \langle #1 , #2 \right \rangle}
\providecommand{\sscpro}[2]{\langle #1 , #2 \rangle}
\providecommand{\bscpro}[2]{\bigl \langle #1 , #2 \bigr \rangle}

\providecommand{\sket}[1]{\vert #1 \rangle}

\providecommand{\sbra}[1]{\langle #1 \vert}

\providecommand{\opro}[2]{\left \vert #1 \right \rangle \left \langle #2 \right \vert}

\providecommand{\dd}{\mathrm{d}}
\providecommand{\id}{\mathrm{id}}

\renewcommand{\Im}{\mathrm{Im} \,}

\providecommand{\ie}{i.~e.~}
\providecommand{\eg}{e.~g.~}

\usepackage[amsmath]{ntheorem}

\theoremstyle{plain}
\newtheorem{thm}{Theorem}[section]
\newtheorem{defn}[thm]{Definition}

\newtheorem{lem}[thm]{Lemma}
\newtheorem{cor}[thm]{Corrolary}
\newtheorem{prop}[thm]{Proposition}
\newtheorem{assumption}[thm]{Assumption}

\theorembodyfont{\upshape}
\newtheorem{remark}[thm]{Remark}

\theoremstyle{nonumberplain}

\theoremsymbol{\ensuremath{_\Box}}
\newtheorem{proof}{Proof}

\usepackage[charter]{mathdesign} 
\SetMathAlphabet{\mathcal}{normal}{OMS}{cmsy}{m}{n} 
\SetMathAlphabet{\mathcal}{bold}{OMS}{cmsy}{m}{n} 
\usepackage{helvet}

\usepackage{epstopdf}
\usepackage{appendix}
\usepackage{wrapfig}
\usepackage{enumerate}
\usepackage{tabularx}

\numberwithin{equation}{section}

\providecommand{\Proof}{\begin{proof}[\textsc{\bf{Proof}}]}
\providecommand{\CVD}{\end{proof}}

\renewcommand{\T}{\mathbb{T}}
\renewcommand{\C}{\mathbb{C}}

\providecommand{\bbb}[1]{\mathcal{#1}}

\providecommand{\ncint}{\mathrel{{\ooalign{$\int$\cr\kern+.07em\raise.15ex\hbox{$\pmb{\scriptstyle-}$}\cr}}}}           
\providecommand{\ncpartial}{\mathrel{{\ooalign{$\partial$\cr\kern+.29em\raise.79ex\hbox{$\pmb{\scriptstyle-}$}\cr}}}}

\newcommand{\ii}{\,{\rm i}\,}				
\def\dd{{\rm d}}
\providecommand{\Fourier}{\mathcal{F}}

\providecommand{\Cont}{\mathcal{C}}

\providecommand{\WS}{\mathcal{W}}
\providecommand{\BZ}{\mathcal{B}^d}
\providecommand{\BF}{\mathcal{U}_{\mathrm{BF}}}
\providecommand{\Usplit}{\mathcal{U}_{\mathrm{split}}}
\providecommand{\spec}{\mathrm{Spec}}
\providecommand{\specrel}{\mathsf{S}}
\providecommand{\Hil}{\mathcal{H}}
\providecommand{\Hrel}{{\Hil_{\specrel}}}
\providecommand{\bundle}{\xi}
\providecommand{\bspace}{\mathcal{E}}
\providecommand{\parity}{\mathcal{P}}
\renewcommand{\ran}{\mathrm{Ran} \, }
\providecommand{\e}{\mathrm{e}}
\providecommand{\vecm}{\mathrm{Vec}_{\C}^m}
\providecommand{\cf}{cf.~}

\theoremstyle{nonumberplain}
\theoremsymbol{}
\newtheorem{question}{Question}

\title{\LARGE Exponentially Localized Wannier Functions\\ in Periodic Zero Flux Magnetic Fields}
\author{$\text{G. De Nittis}^*$ \& $\text{M. Lein}^{**}$}
\date{\today}

\begin{document}

\maketitle              
\vspace{-9mm}
\begin{center}
	$^\ast$ LAGA, Institut Galil\'{e}e, Universit\'{e} Paris 13 \linebreak
	99, avenue J.-B. Cl\'{e}ment, F-93430 Villetaneuse, France. \linebreak
	{\footnotesize \texttt{denittis@math.univ-paris13.fr}}\linebreak
	\linebreak
	$^{\ast\ast}$ Eberhard Karls Universität Tübingen, Mathematisches Institut \linebreak
	Auf der Morgenstelle 10, 72076 Tübingen, Germany. \linebreak
	{\footnotesize \texttt{lein@ma.tum.de}}
\end{center}
\thispagestyle{empty}          
\begin{abstract}
	In this work, we investigate conditions which ensure the existence of an exponentially localized Wannier basis for a given periodic hamiltonian. We extend previous results in \cite{Panati:triviality_Bloch_bundle:2006} to include periodic zero flux magnetic fields which is the setting also investigated in \cite{Kuchment:exponential_decaying_wannier:2009}.
	The new notion of \emph{magnetic symmetry} plays a crucial rôle; to a large class of symmetries for a non-magnetic system, one can associate ``magnetic'' symmetries of the related magnetic system. 
	Observing that the existence of an exponentially localized Wannier basis is equivalent to the triviality of the so-called Bloch bundle, a rank $m$ hermitian vector bundle over the Brillouin zone, we prove that \emph{magnetic} time-reversal symmetry is sufficient to ensure the triviality of the Bloch bundle in spatial dimension $d=1,2,3$.
	For $d = 4$,  an exponentially localized  Wannier basis exists provided that the \emph{trace per unit volume} of a suitable function of the Fermi projection vanishes.
	For $d>4$ and $d\leqslant 2m$ (stable rank regime) only the exponential localization of a subset of Wannier functions is shown; this improves part of the analysis of \cite{Kuchment:exponential_decaying_wannier:2009}.
	Finally, for $d>4$ and $d> 2m$ (unstable rank regime) we show that the mere analysis of Chern classes does not suffice in order to prove trivility and thus exponential localization.
\end{abstract}

\tableofcontents

\pagestyle{headings}

\section{Introduction} 
\label{intro}
A simple, but standard quantum model to study the conductance properties of a crystalline solid is provided by the one-particle hamiltonian 
\begin{align}\label{intro:eqn:HA}
	H^A = \bigl ( - \ii \nabla_x - A(\hat{x}) \bigr )^2 + V(\hat{x})
\end{align}
on $L^2(\R^d)$ where $V$ and $A$ are effective electric and magnetic potentials generated by the nuclei and all other electrons. The vector potential $A = (A_1 , \ldots , A_d)$ is associated to the magnetic field $B: = (B_{ij})$ with components $B_{ij} := \partial_{x_i} A_j - \partial_{x_j} A_i$ and $1 \leqslant i,j \leqslant d$. We will always assume the functions $V$, $A$ and consequently $B$ are periodic with respect to the lattice 
\begin{align*}
	\Gamma := \Bigl \{ \gamma \in \R^d \; \; \big \vert \; \; \mbox{$\gamma = \sum_{j = 1}^d n_j \, e_j$, } \quad n_j \in \Z, \; j = 1 , \ldots , d \Bigr \} 
\end{align*}
where $\{e_1 , \ldots , e_d \}$ is a basis of $\R^d$. Throughout the paper, we will make the following 
\begin{assumption}\label{bloch_floquet:assumption:HA}
	\begin{enumerate}[(i)]
		\item $V \in L^1_{\mathrm{loc}}(\R^d)$ is $\Gamma$-periodic. 
		\item The components of $B$ are $\Gamma$-periodic and $\Cont^1$. Furthermore, $B$ satisfies the following \emph{zero flux} condition: for all $1 \leqslant j < l \leqslant d$
		\begin{align}
			\int_{C_{jl}} B = 0 
			\label{bloch_floquet:eqn:zero_flux_condition}
		\end{align}
		holds where $C_{jl}$ denotes the parallelogram with vertices $0$, $e_j$, $e_j + e_l$ and $e_l$. 
	\end{enumerate}
\end{assumption}
\begin{remark}[Vector potentials]\label{rem:vec_pot}
	The assumptions obviously include the case $B = 0$. To any magnetic field satisfying (ii), we can associate a $\Gamma$-periodic vector potential $A$ such that $\dd A = B$ \cite[Proposition 1]{Hempel_Herbst:band_gaps_periodic_mag_hamiltonians:1994}. \emph{In fact, unless explicitly stated otherwise, we will always assume the vector potential of choice $A$ is $\Gamma$-periodic.} 
\end{remark}
Under these assumptions, $H^A$ defines a selfadjoint operator on a suitable domain in $L^2(\R^d)$. 
Moreover, its spectrum is absolutely continuous \cite{Suslina:ac_spectra_periodic_operators:2000} and has a \emph{band structure}, \ie it can be seen as the locally finite union of closed intervals called \emph{bands} which are separated from each other by \emph{gaps} \cite{Hempel_Herbst:band_gaps_periodic_mag_hamiltonians:1994,Bruening_Sunada:gauge-periodic_elliptic_operators:1992,Bruening_Sunada:periodic_elliptic_operators:1992,Kuchment:Floquet_theory:1993}. 
\medskip

\noindent
Any efficient description of physical properties of crystalline solids makes use of the translational symmetry. In quantum mechanics, ``efficient'' often means two things: (1) one focusses on a ``relevant'' energy regime and (2) one chooses a ``good'' basis for the states that lie in the energy range of interest. In solid state physics, the relevant energy range is typically the energies below the Fermi energy.  There are two standard choices of a basis: \emph{Bloch functions} are more appropriate to describe delocalized phenomena such as conduction as they are the analogs of plane waves \cite{Wilcox:theory_Bloch_waves:1978}. Or one can use \emph{Wannier functions} which are the analog of Dirac $\delta$ functions or molecular orbitals and are typically well-localized in a region around a specific fundamental cell \cite{Marzari_Vanderbilt:maximally_localized_Wannier_functions:1997,Marzari_Souza_Vanderbilt:maximally_localized_Wannier_functions:2003}. Hence, Wannier functions are more appropriate to justify tight-binding models, for instance \cite{Ashcroft_Mermin:solid_state_physics:2001,Grosso_Parravicini:solid_state_physics:2003,Wannier:Wannier_functions:1937}. 

A characteristic of gapped semiconductors and insulators is the existence of gaps which separate the relevant part of the spectrum $\specrel \subset \spec(H^A)$ from the remainder, namely 
\begin{align}
	\mathrm{dist} \bigl ( \specrel , \spec(H^A) \setminus \specrel \bigr ) > 0
	.
	\label{eq:G1}
\end{align}
If the Fermi energy $E_{\mathrm{F}}$ lies in a spectral gap, then the material is an insulator or a semiconductor depending on the width of the gap. In this case, one generally focusses on the portion of the spectrum that lies below $E_{\mathrm{F}}$, namely one chooses $\specrel = \spec(H^A) \cap (-\infty,E_{\mathrm{F}}]$ as relevant part of the spectrum. Later on, we will see that a \emph{global} gap in the sense of equation~\eqref{eq:G1} is not necessary, a \emph{local} gap as in Assumption~\ref{bloch_floquet:assumption:spectral_gap} suffices. 

Since $\spec(H^A)$ consists of closed intervals, $\specrel$ is also the union of closed intervals. Due to the spectral gap, the spectral projection 
\begin{align}
	P_{\specrel} := 1_{\specrel}(H^A) 
	= \frac{\ii}{2 \pi} \int_{\mathsf{C}} \dd \zeta \, \bigl ( H^A - \zeta \bigr )^{-1} 
\end{align}
can be written either in terms of the characteristic function $1_{\specrel}$ associated to the Borel set $\specrel \subset \C$ or as a Cauchy integral involving a countour $\mathsf{C} \subset \C$ enclosing $\specrel$. To simplify notation, we have suppressed the dependence of $P_{\specrel}$ on the vector potential $A$. Associated to $\specrel$, we can define a special family of basis functions: 
\begin{defn}[Wannier system]\label{intro:defn:Wannier_system}
	A Wannier system $\{ w_1 , \ldots , w_m \} \subset L^2(\R^d)$ associated to the projection $P_{\specrel}$ is a family of orthonormal functions so that their translates $w_{j,\gamma}(\cdot) := w_j(\cdot - \gamma)$ are mutually orthonormal, $\sscpro{w_{j,\gamma}}{w_{j',\gamma'}} = \delta_{j j'} \, \delta_{\gamma \gamma'}$, and one can write the projection as 
	\begin{align*}
		P_{\specrel} &= \sum_{j = 1}^m \sum_{\gamma \in \Gamma} \sket{w_{j,\gamma}}\sbra{w_{j,\gamma}}
		. 
	\end{align*}
	The functions $w_j$ are known as \emph{Wannier functions} and the integer $m$ is referred to as \emph{geometric rank of $P_{\specrel}$.} 
\end{defn}
The geometric rank $m$ of any Wannier system associated to $\specrel$ turns out to be independent of the particular choice of Wannier system (cf.~Section~\ref{bloch_floquet:wannier_functions}). 

Numerically, one would like to work with ``quickly decaying'' Wannier functions since then, many computational methods scale linearly with system size \cite{Marzari_Vanderbilt:maximally_localized_Wannier_functions:1997,Marzari_Souza_Vanderbilt:maximally_localized_Wannier_functions:2003,wannier.org:2011}. 
As we will see later on, the relevant notion is that of exponential decay in a $L^2$ sense, namely:%
\begin{defn}[Exponential decay]\label{dfn:ex_dec}
	Let $K \subset \R^d$ be any compact subsetset,  $ K - \gamma$ its translate by $\gamma \in \Gamma$ and $1_{K-\gamma}$ the associated characteristic function. We say that $w \in L^2(\R^d)$ is exponentially localized if and only if there exists a $b \in (0,+\infty)$ such that 
	\begin{align*}
		\sup_{\gamma \in \Gamma} \bnorm{1_{K-\gamma}\ w}_{L^2(\R^d)} \, \e^{b \abs{\gamma}} < \infty 
	\end{align*}
	holds. 
\end{defn}
Obviously, this definition does not depend on the particular choice of $K$ and we may use the Wigner-Seitz cell $\WS = K$ (cf.~Section~\ref{bloch_floquet}), for instance. 
\medskip

\noindent
This paper concerns itself with the following 
\begin{question}
	Under which conditions on the spatial dimension $d$ and the geometric rank $m$ of $P_{\specrel}$ can we guarantee the existence of an exponentially localized Wannier system?
\end{question}
We will adopt the commonly used strategy which is to rephrase this question in terms of Bloch functions and make use of the rich tool set provided by complex analysis and fiber bundle theory. 
\medskip

\noindent
Let us first recall the connection between Wannier and Bloch functions (further details can be found in Section~\ref{bloch_floquet}). The lattice symmetry of $H^A$ allows us to fiber decompose the hamiltonian in terms of Bloch momentum $k$ by means of the Bloch-Floquet transform $\BF$, and one obtains a family of operators 
\begin{align*}
	H^A(k) = \bigl ( - \ii \nabla_y - A(\hat{y}) \bigr )^2 + V(\hat{y}) 
\end{align*}
on $L^2(\WS)$ with $k$-dependent domains. Here $\WS$ is the Wigner-Seitz cell and $\hat{y}$ denotes the position operator. For each $k \in \R^d$, the operator $H^A(k)$ has compact resolvent and thus its spectrum consists of eigenvalues that accumulate at infinity. A Bloch function $\varphi_n^A(k)$ is then the solution to the eigenvalue equation 
\begin{align*}
	H^A(k) \varphi_n^A(k) &= E_n^B(k) \, \varphi_n^A(k)
\end{align*}
where $E_n^B(k)$ is the $n$-th eigenvalue. As is customary, we will order the $E_n^B(k)$ by magnitude and the functions $k \mapsto E_n^B(k)$ are called \emph{Bloch bands.} As the symbol $E_n^B$ indicates, each Bloch band depends only on the magnetic field. This is a direct consequence of the \emph{gauge covariance} of $H^A(k)$ (cf.~Section~\ref{bloch_floquet:gauge_covariance}). Since $H^A(k)$ is analytic in $k$, the energy band functions $k \mapsto E_n^B(k)$ are analytic away from band crossings. Furthermore, since for each $k$ we are free to multiply $\varphi_n^A(k)$ by a phase, we can choose Bloch functions so that the $k \mapsto \varphi_n^A(k)$ are piecewise analytic \cite{Wilcox:theory_Bloch_waves:1978}. 

Now we can set 
\begin{align}
	w_n^A(x) := (\BF^{-1} \varphi_n^A)(x) := \int_{\BZ} \dd k \, \e^{- \ii k \cdot (x - [x]_{\WS})} \, \varphi_n^A(k,[x]_{\WS}) 
	\label{intro:eqn:wannier_function}
\end{align}
where we have split  $x = [x]_{\WS} + (x - [x]_{\WS})$ into a contribution contained within the fundamental cell $[x]_{\WS} \in \WS$ and a lattice vector $x - [x]_{\WS} \in \Gamma$. The domain $\BZ$ is the usual \emph{Brillouin zone} which is the fundamental cell of the \emph{dual lattice}  $\Gamma^\ast$ whose faces have been identified (cf.~Section~\ref{bloch_floquet:bloch_floquet_transform}) and $\dd k$ denotes the normalized Lebesgue measure on $\BZ$. One can show that $\{ w_1^A , \ldots,  w_m^A \}$ forms a Wannier system in the sense of Definition~\ref{intro:defn:Wannier_system} and the geometric rank $m$ of $P_{\specrel}$ is the number of Bloch bands including multiplicity which make up $\specrel$. Since the $\{ w_1^A,\ldots,  w_m^A \}$ are essentially Fourier transforms of Bloch functions, regularity properties of Bloch functions translate into decay properties of Wannier functions. If the Bloch function $\varphi_n^A(k)$ associated to the $n$-th band is analytic, then by standard arguments (a variant of the \emph{Paley-Wiener Theorem}) this translates into exponential decay in the sense of Definition~\ref{dfn:ex_dec}. 

A priori, it is not even clear whether we can make a choice of phase so that Bloch functions $\varphi_n^A$ are  continuous on all of $\R^d$ \emph{and} $\Gamma^\ast$-periodic in $k$. Even if one were to assume continuity and $\Gamma^\ast$-periodicity, we know that in general, we cannot expect the Bloch functions to be more regular as they are only continuous at band crossings. 
Hence, the associated Wannier functions $w_n^A$ typically decay slowly, a consequence of the discontinuities of $k \mapsto \varphi_n^A(k)$ or its derivatives. However, there still may be another family of analytic and $\Gamma^*$-periodic functions 
\begin{align*}
	\Bigl \{ \psi_j : \BZ \longrightarrow L^2(\WS) \; \big \vert \; j = 1 , \ldots , m \Bigr \}
\end{align*}
such that for each $k$, they are an orthonormal basis of 
\begin{align}
	\Hil_{\specrel}(k) := \ran P_{\specrel}(k)
	= \mathrm{span} \, \bigl \{ \psi_1(k) , \ldots , \psi_m(k) \bigr \} 
	\subset L^2(\WS)
	\label{intro:eqn:fiber_space}
	. 
\end{align}
Then by virtue of the Paley-Wiener theorem, the associated Wannier system $\{ w_1,\ldots,w_m \}$, $ w_j:= \BF^{-1} \psi_j$, consists of exponentially decaying functions. 
The dimension $m$ of the spaces $\Hil_{\specrel}(k)$ is independent of $k$ due to the analyticity of the spectral projection $P_{\specrel}(k)$ (given by equation~\eqref{bloch_floquet:eqn:P_k}) and coincides with the geometric rank of $P_{\specrel}$ given in Definition~\ref{intro:defn:Wannier_system}. 
\medskip

\noindent
Before we proceed, let us discuss whether some conditions in the problem of proving the existence of localized Wannier functions can be relaxed or whether asking for exponential decay is the only sensible localization criterion. In principle, one can think of two ways to simplify the question: one can either give up exponential decay or orthonormality. Let us start with decay: one may think that the requirement of exponential decay can be relaxed to polynomial decay, \ie 
\begin{align*}
	\bnorm{1_{\WS - \gamma} \, w_j}_{L^2(\R^d)} \leqslant C (1 + \sabs{\gamma})^{-\beta}
	,
	&&
	\forall \; j = 1 , \ldots , m
	, \;
	\gamma \in \Gamma
	, 
\end{align*}
holds for some $C > 0$ and $\beta > 0$. But interestingly, Kuchment has shown \cite[Theorem~5.4]{Kuchment:exponential_decaying_wannier:2009} that this is not the case: if one can show that the $w_j$ are polynomially localized for $\beta > d$, then one can find a Wannier system which is exponentially localized. If on the other hand $\beta \leqslant d$, then the decay is considered to be too slow to be useful for numerics. 

A second option would be to give up the requirement that the $w_j$ be mutually orthogonal and one works with an overcomplete set of exponentially localized functions. Kuchment has shown that it is always possible to find a family of exponentially localized functions $\{ \tilde{w}_1 , \ldots , \tilde{w}_{m'} \}$ for some $m \leqslant m' \leqslant 2^d m$ such that their translates $\tilde{w}_{j , \gamma} := \tilde{w}_j(\cdot - \gamma)$ span $\ran P_{\specrel}$ \cite[Theorem~5.7]{Kuchment:exponential_decaying_wannier:2009}. So even in cases, where we are unable to find an exponentially localized \emph{basis}, one can still choose $m' > m$ exponentially localized functions $\tilde{w}_j$ that form a so-called $1$-tight (or Parseval) frame rather than an orthonormal basis. Parseval frames, even though they are overcomplete and do not consist of orthonormal functions, retain many of the advantages of a true orthonormal basis. 

\subsubsection*{State of the art} 
\label{intro:previous_works}
Historically, the question of existence of exponentially localized Wannier functions has received a lot of attention, especially in the absence of magnetic fields. 

Let us focus on the results concerning the case $B = 0$. The first rigorous work by Kohn dates back to 1959 \cite{Kohn:analytic_properties_Bloch_Wannier:1959} and treats the $d = 1$ case. Due to the dimensionality, it is always possible to choose analytic Bloch functions and thus find exponentially decaying Wannier functions. Note that since there are no magnetic fields, Kohn's result covers all possibilities in $d = 1$.

Since Kohn's ideas do not extend to higher spatial dimensions, it took until 1983 to see progress for the case $d \geqslant 2$ and $m = 1$: Nenciu \cite{Nenciu:exponential_loc_Wannier:1983} used \emph{time-reversal symmetry} to prove the existence of exponentially localized Wannier functions associated to isolated bands. An independent and elegant constructive proof was given by Helffer and Sjöstrand in 1989 \cite{Helffer_Sjoestrand:mag_Schroedinger_equation:1989}. 

The first to realize the rôle of topological obstructions in proving existence of exponentially localized Wannier functions was Thouless \cite{Thouless:Wannier_functions_mag_subbands:1984}. He noticed that it is not possible to choose exponentially localized Wannier functions if the first Chern class of an associated vector bundle (\ie the \emph{Bloch bundle}, cf.~Section~\ref{bloch_bundle}) does not vanish. Simon \cite{Simon:holonomy_Berrys_phase:1983} had recognized that one such condition which guarantees the vanishing of the first Chern class is the absence of magnetic fields, \ie in the presence of time-reversal symmetry. Nenciu \cite{Nenciu:effective_dynamics_Bloch:1991} stressed the significance of the \emph{Oka principle} (cf.~Section~\ref{bloch_floquet:the_oka_principle}) which meant that proving the existence of \emph{continuous} $\Gamma^*$-periodic Bloch functions implied the existence of \emph{analytic} $\Gamma^*$-periodic Bloch functions. 

In the works of Panati \cite{Panati:triviality_Bloch_bundle:2006}, all these pieces are put together and the triviality of the Bloch bundle for $m \geqslant 1$ and $d \leqslant 3$ is shown, thus covering most physical situations. He shows that the vanishing of the first Chern class suffices to ensure the triviality of the Bloch bundle for such low dimensions  $d$. The link to exponential localization of Wannier functions was explained in a subsequent publication \cite{brouder-panati-et.al-07}. 
\medskip

\noindent
The geometric content of Panati's proof is tied to an important result from the classification theory of vector bundles by Peterson \cite{peterson-59}: it assures  that the vanishing of all Chern classes is a necessary and sufficient condition for the triviality of the Bloch bundles only if $d \leqslant 2 m$, \ie if the fibers are of high enough dimension; this is called the \emph{stable rank condition.} Since only Chern classes $\mathrm{c}_j$ for which $2j \leqslant d$ holds can be non-trivial, for $d = 2 , 3$, one only needs to investigate the first Chern class. By the fortunate coincidence that line bundles, \ie $m = 1$, are completely characterized by their first Chern class, and that for $m \geqslant 2$, the stable rank condition $d \leqslant 2m$ is always satisfied as long as $d \leqslant 3$, the vanishing of the first Chern class is equivalent to the triviality of the bundle. Thus, Panati's result is included as a special case of a deeper fact from the classification theory of vector bundles. 

The above analysis also justifies the absence of any result in higher dimensions. Indeed, to show exponential localization for the relevant case $d = 4$ and $m \geqslant 2$, for instance, one needs to control the second Chern class, and time-reversal symmetry is of no help here (\cf Theorem~\ref{geo_analysis:thm:triviality_odd_c_j}). 
Moreover, if $d \geqslant 5$ and the fibers are low-dimensional, $d > 2m$, then there are examples of non-trivial vector bundles whose Chern classes all vanish (cf.~Section~\ref{geom_analysis:limits_characteristic_classes}). In this sense, the stable rank condition is not a mere technical, but an essential condition, and Peterson's result shows the limitations of using characteristic classes to prove exponential localization of Wannier functions. 
\medskip

\noindent
For the magnetic case $B\neq0$, the literature is more scarce: there are early works by Dubrovin and Novikov \cite{Dubrovin_Novikov:2d_particle_periodic_B:1980} and Novikov \cite{Novikov:magnetic_Bloch_functions_vector_bundles:1980} which treated the case of periodic magnetic field with rational flux perturbed by a weak potential. They recognize that properties of magnetic Bloch functions are determined by the geometry of the Bloch bundle, and that in general, the presence of the magnetic fields makes this vector bundle topologically non-trivial. This is in accord with the previously mentioned paper by Thouless \cite{Thouless:Wannier_functions_mag_subbands:1984}. 

Without making use of bundle theory, Nenciu proved the existence of exponentially localized Wannier functions for a weak constant magnetic field \cite[Theorem~5.2]{Nenciu:effective_dynamics_Bloch:1991}: he has shown that for a single isolated band, the projection associated to magnetic translations of the exponentially localized \emph{non}-magnetic Wannier function is close to the true spectral projection for $B \neq 0$ and then proceeds to prove the existence of an exponentially localized \emph{magnetic} Wannier function by a perturbative argument. His result is not in contradiction to the results mentioned earlier: the weak field condition ensures that one is in the regime where the first Chern class is zero. 

Lastly, we would like to mention again the important work by Kuchment \cite{Kuchment:exponential_decaying_wannier:2009} which is complementary: for magnetic fields which admit periodic vector potentials (our setting), he does not show the triviality of the Bloch bundle, but instead that if one gives up the orthonormality of the Wannier functions: one can always find an overcomplete set of exponentially localized Wannier functions. Indeed, as we shall see in a moment, if the spatial dimension is small enough, it is still possible to have a genuine system of $m$ exponentially localized Wannier functions. 

\subsubsection*{New results: magnetic symmetries and geometric conditions on triviality}
The presence of magnetic fields in general changes the topology of the Bloch bundle. In this work though, we will show that for magnetic fields which admit periodic vector potentials, the topology of the Bloch bundle is the same as in the case when $B = 0$. 

In the absence of magnetic fields, time-reversal symmetry is at the heart of most proofs which show the existence of exponentially localized Wannier functions. While magnetic fields break time-reversal symmetry, we show the hamiltonian $H^A$ has \emph{magnetic} time-reversal symmetry. 
\begin{thm}[Magnetic time-reversal symmetry]\label{intro:thm:mag_time-reversal_symmetry}
    Let Assumption~\ref{bloch_floquet:assumption:HA} on $H^A$ be satisfied. Then the operator $H^A$ given by \eqref{intro:eqn:HA} has magnetic time-reversal symmetry, that is, it commutes with the anti-unitary operator $C^A$ defined by 
    \begin{align}
        (C^A \Psi)(x) := \e^{+ \ii 2 \int_{[0,x]} A} \, \Psi^*(x)
    \end{align}
    where $\Psi \in L^2(\R^d)$ and $x \in \R^d$. Furthermore, magnetic time-reversal commutes with lattice translations. 
\end{thm}
The idea of ``magnetic symmetries'' can be put in a larger context: if $H^0 := H^{A = 0}$ has a symmetry with certain natural properties (cf.~Definition~\ref{mag_sym:defn:S-transform}), then we show how to associate a magnetic symmetry to $H^A$ in a canonical way (cf.~Definition~\ref{mag_sym:defn:mag_S-transform}). Gallilean symmetries (\eg translations and rotations), for instance, have a magnetic version. We reckon that magnetic symmetries could be useful in different contexts and this direction will be explored in a future work \cite{DeNittis_Lein:magnetic_symmetries:2011}. 

The presence of this newly found symmetry $C^A$ of $H^A$ is the crucial ingredient to prove 
\begin{thm}[Existence of an exponentially localized Wannier system]\label{intro:thm:existence_exp_loc_Wannier}
    Let Assump-\linebreak tions~\ref{bloch_floquet:assumption:HA} and \ref{bloch_floquet:assumption:spectral_gap} be satisfied, and let us denote the dimension of $\BZ$ and the geometric rank of the Wannier system with $d$ and $m$. Then an exponentially localized Wannier system associated to $P_{\specrel}$ exists if 
    \begin{enumerate}[(i)]
        \item $d \geqslant 1$ and $m = 1$ or 
        \item $d = 1 , 2, 3$ and $m \geqslant 2$. 
    \end{enumerate}
\end{thm}
This theorem is proven by combining Proposition~\ref{bloch_bundle:prop:main_question_bundle_language} with Proposition~\ref{geom_analysis:prop:triviality_m_1} for (i) and Corollary~\ref{geo_analysis:cor:low_d} for (ii), respectively. Our proof is very close in spirit to Panati's \cite{Panati:triviality_Bloch_bundle:2006}, but there are a few crucial differences: we use more abstract arguments, \eg cohomology classes and the Kronecker pairing rather than differential forms and integration. This gives us access to powerful tools of algebraic topology for the classification of vector bundles. 

In this language, the presence of time-reversal symmetry leads to a geometric constraint connecting the fibers attached to conjugate points (\cf Theorem~\ref{mag_symm:thm:consequence_mag_symmetry_bundle_geometry}), and the \emph{conjugate} Bloch bundle can be seen as the pullback bundle with respect to the function which maps $k \mapsto -k$. Since Chern classes of a bundle and its conjugate agree up to a sign, it follows that \emph{all} odd Chern classes vanish (\cf Theorem~\ref{geo_analysis:thm:triviality_odd_c_j}). Note that the dimensional constraints are necessary since in higher dimensions $d$, additional topological obstructions (measured by Chern classes, cf.~Section~\ref{geom_analysis:characteristic_classes_and_triviality}) appear and the presence of time-reversal symmetry does \emph{not} suffice to ensure the existence of an exponentially localized Wannier system. 

Apart from being of pure mathematical interest, the case $d = 4$ is also relevant in many physical situations: for instance, if one considers periodic deformations of period $T$ of crystalline solids in $3$ spatial dimensions, then space-time is given by $\mathcal{B}^3 \times  \R / T \Z \cong \T^4$ \cite{Lein:polarization:2005,Panati_Sparber_Teufel:polarization:2006}. In this context,
Section~6 of \cite{Panati_Sparber_Teufel:polarization:2006} which covers the rôle of time-reversal symmetry and parity should be compared to Theorem~\ref{mag_symm:thm:consequence_mag_symmetry_bundle_geometry} and its ramifications on the structure of the Bloch bundle (\cf beginning of Section~\ref{geo_analysis:odd_c_j_zero}). Moreover, the geometric aspects related to the case $d=4$ seem to have interesting connections with gauge theories on ``compactified space-times'' $\T^4$  \cite{Nash:gauge_potentials_4_torus:1983,Schenk:generalized_Fourier_trafo_instantons:1988,Braam_Van_Baal:Nahm_trafo_instantons:1989}. Albeit in $d=4$ it is not possible to prove \emph{a priori} the existence of an exponentially localized Wannier system if $m \geqslant 2$, our analysis yields a 
\begin{thm}[Criterion for exponential localization in $d = 4$]\label{intro:thm:d_4_criterion}
	Let Assumptions~\ref{bloch_floquet:assumption:HA} and \ref{bloch_floquet:assumption:spectral_gap} be satisfied, and assume $d = 4$ and $m \geqslant 2$. Then there exists an exponentially localized Wannier system associated to $P_{\specrel}$ if and only if 
	\begin{align}
	    \mathcal{T}(W_{\specrel}) = 0
	    \label{intro:eqn:d_4_criterion}
	\end{align}
	where $\mathcal{T}$ denotes the \emph{trace per unit volume} (\cf equation \eqref{eq:TPUV}) and $W_{\specrel}$ is related to $P_{\specrel}$ by
	\begin{align*}
		W_{\specrel} := Q_{12}(P_{\specrel}) \, Q_{34}(P_{\specrel}) -  Q_{13}(P_{\specrel}) \, Q_{24}(P_{\specrel}) + Q_{14}(P_{\specrel}) \, Q_{23}(P_{\specrel}) 
	\end{align*}
	with $Q_{ij}(P_{\specrel}) := P_{\specrel} \, \bigl [ \delta_i P_{\specrel} , \delta_j P_{\specrel} \bigr ] \, P_{\specrel}$ and $\delta_j := - \tfrac{\ii}{2\pi} [\hat{x}_j , \cdot]$ the \emph{$j$-th derivative} (\cf equation~\eqref{eq:j-th_der}).
\end{thm}
This result is an immediate corollary of Theorem \ref{theo:CC_d4} and Proposition \ref{prop:trace_deriv}. The main point is that the bounded operator $W_{\specrel}$ is related to the second differential Chern class of the Bloch bundle and condition~\eqref{intro:eqn:d_4_criterion} implies the vanishing of the latter.
Prodan has proposed an efficient numerical scheme to evaluate quantities like $\mathcal{T}(W_{\specrel})$ \cite{Prodan:topological_insulators:2011}, so the numerical verification of equation~\eqref{intro:eqn:d_4_criterion} provides a criterion to decide whether or not an exponentially localized Wannier system exists in $d=4$. 

Apart from proving the existence of an exponentially localized Wannier system, our results can be applied to space-adiabatic perturbation theory \cite{PST:effDynamics:2003,DeNittis_Lein:Bloch_electron:2009,DeNittis_Panati:effective_models_conductance:2010}: here, one needs a smooth Bloch basis in order to satisfy a technical assumption in the construction (Assumption~$\mathrm{A}_2$ in \cite{PST:effDynamics:2003} or Assumption~3.2 in \cite{DeNittis_Lein:Bloch_electron:2009}). Our result here allows us to extend \cite{DeNittis_Lein:Bloch_electron:2009} to include periodic magnetic fields: any $\Gamma$-periodic magnetic field $B = B_{\mathrm{flux}} + B_0$ can be decomposed into a constant magnetic field $B_{\mathrm{flux}}$ whose flux through the Wigner Seitz cell $\WS$ coincides with the total flux of $B$ and a magnetic field $B_0$ with zero flux through $\WS$. If the total flux of $B$ is small, then $B_{\mathrm{flux}}$ is small and one can regard $B_{\mathrm{flux}}$ as a perturbation of $B_0$. The zero flux field enters the unperturbed hamiltonian by a $\Gamma$-periodic vector potential while the constant field is seen as a perturbation. Thus, one can repeat the arguments in \cite{DeNittis_Lein:Bloch_electron:2009} verbatim, but in this case the unperturbed objects (symbols, opertators, etc.) include the periodic magnetic field and vector potential. 
\medskip

\noindent
In our approach, it is clear how to generalize Theorem~\ref{intro:thm:existence_exp_loc_Wannier} to $d > 4$: in view of Corollary~\ref{geo_analysis:cor:triviality_Bloch_bundle_Peterson} and Theorem~\ref{geo_analysis:thm:triviality_odd_c_j}, an exponentially localized Wannier system exists, provided that $\lfloor \nicefrac{d}{4} \rfloor$ (\ie the \emph{integer part} of  $\nicefrac{d}{4}$) extra conditions are verified, namely the vanishing of the even Chern classes. Otherwise, we can only ensure the exponential localization of some of the Wannier functions. 
\begin{thm}[Partially localized Wannier systems]\label{intro:thm:partial_loc_Wannier}
    Let Assumptions~\ref{bloch_floquet:assumption:HA} and \ref{bloch_floquet:assumption:spectral_gap} be satisfied. Define $\sigma := \max \bigl \{ 0 , m - \lfloor \nicefrac{d}{2} \rfloor \bigr \}$. Then the following statements hold true: 
    \begin{enumerate}[(i)]
        \item At least $\sigma$ generators of the Wannier system are exponentially localized. 
        \item In the special case $d = 4k+2$, at least $\sigma + 1$ generators of the Wannier system are exponentially localized. 
    \end{enumerate}
\end{thm}
This result follows from Proposition~\ref{bloch_bundle:prop:main_question_bundle_language} and some general facts concerning the classification of vector bundles. The proof will be given in Section~\ref{geom_analysis:partial_loc}. With the above theorem, we can improve the results of Kuchment: we can give a smaller upper bound on the number $m'$ of functions spanning the $1$-tight frame, namely $m \leqslant m' \leqslant M$ where $M = 2^d (m - \sigma) + \sigma$
or $M = 2^d (m -  \sigma - 1) + \sigma + 1$ in the special case $d = 4 k + 2$. 

Finally, if $d > 2m$, then the vanishing of all Chern classes is only a necessary but not sufficient condition for the triviality of the Bloch bundle. Indeed, in Section~\ref{geom_analysis:limits_characteristic_classes} we will construct a rank $2$ bundle over $\mathcal{B}^5$ which is non-trivial but whose Chern classes all vanish. This means that in the unstable rank regime, one needs to complement the analysis of characteristic classes with other techniques in order to prove the existence of an exponentially localized Wannier system. 

\subsubsection*{Content} 
\label{intro:content}
The paper is organized as follows: we first give a brief introduction to Bloch-Floquet theory in \textbf{Section~\ref{bloch_floquet}} and discuss how the Oka principle ties in with the main result. The trace per unit volume and its relation with the Bloch-Floquet decomposition will be discussed.
The notion of magnetic symmetry is given in \textbf{Section~\ref{sec_magn_sym}}. 
We cover not only the important case of magnetic time-reversal and parity, but also give other examples of other Galilean symmetries. 
In \textbf{Section~\ref{bloch_bundle}} we give a primer on vector bundle theory and construct the Bloch bundle. 
We explain the link between exponential localization of the Wannier system and the triviality of the Bloch bundle. Moreover, the consequences of the presence of magnetic time-reversal symmetry or parity on the structure of the Bloch bundle are explored.
Finally, in \textbf{Section~\ref{geo_analysis}} the geometry of the Bloch bundle is characterized by the Chern classes. In particular, we prove the vanishing of all odd Chern classes (not only the first!)
and we deduce conditions for the triviality of the Bloch bundle. 

\subsubsection*{Acknowledgements} 
\label{intro:acknowledgements}
G.~D.~is supported by the grant ANR-08-BLAN-0261-01. M.~L.~would like to thank the German-Israeli Foundation for kind support. Furthermore, G.~D.~and M.~L. thank the Erwin Schrödinger Institute for its kind hospitality
during the final drafting of this manuscript. We are also grateful to C.~Rojas Molina, S.~Teufel and P.~Kuchment for their helpful comments. 


\section{The Bloch-Floquet theory} 
\label{bloch_floquet}
To fix notation and make this work self-contained, we will give a short overview over Bloch-Floquet theory starting from a slightly non-standard angle which allows for interesting generalizations \cite{Bellissard_De_Nittis_Milani:aperiodic_Wannier:2011}.
For a more detailed account, we refer to  \cite{Wilcox:theory_Bloch_waves:1978,Reed_Simon:M_cap_Phi_4:1978,Berezin_Shubin:Schroeinger_eq:1991,Kuchment:Floquet_theory:1993}. 

\subsection{The Bloch-Floquet transform} 
\label{bloch_floquet:bloch_floquet_transform}
Since $H^A$ is $\Gamma$-periodic, we will first decompose $\R^d$ into $\Gamma \times \WS$, \ie we write each $\R^d \ni x = \gamma + y$ where $\gamma \in \Gamma$ is a lattice vector and $y \in \WS$ is a vector attached to the origin of the \emph{fundamental cell} $\WS$. In physics, one often chooses the \emph{Wigner-Seitz cell} (also known as \emph{Varanoi cell}), but any other open convex polytope $\WS$ satisfying (i) $0 \in \WS$, (ii) $\bigcup_{\gamma \in \Gamma} \overline{\WS + \gamma} = \R^d$ and (iii) $(\WS + \gamma_1) \cap (\WS + \gamma_2) = \emptyset$ for all $\gamma_1 \neq \gamma_2$ will do. The family $\{ \overline{\WS + \gamma} \}_{\gamma \in \Gamma}$ defines a periodic tiling of $\R^d$. 

The splitting $\R^d \cong \Gamma \times \WS$ induces the decomposition 
\begin{align*}
	\Usplit : L^2(\R^d) \longrightarrow \ell^2(\Gamma) \otimes L^2(\WS)
	, 
	\qquad\quad
	\Usplit \Psi := \sum_{\gamma \in \Gamma} \delta_{\gamma} \otimes T_{-\gamma} \Psi \vert_{\WS} 
	, 
\end{align*}
where $\{ \delta_{\gamma} \}_{\gamma \in \Gamma}$ is the canonical basis of $\ell^2(\Gamma)$ and $T_{\gamma}$ denotes lattice translations by $\gamma \in \Gamma$, $(T_{\gamma} \Psi)(x) := \Psi(x - \gamma)$. The splitting of real space induces a splitting of momenta $p \in {\R^d}^* \cong \Gamma^* \times \BZ$ into a dual lattice vector $\gamma^* \in \Gamma^*$ and an element of the first Brillouin zone $k \in \BZ$. The \emph{dual lattice} is spanned by the vectors $\{ e_1^* , \ldots , e_d^* \}$ which are defined through the relation $e_l \cdot e_j^* = 2 \pi \delta_{lj}$ while the \emph{Brillouin zone} is the quotient group $\BZ := {{\R^d}^*} / {\Gamma^*}$. Note that this definition deviates from the more common one where $\BZ$ is a fundamental cell associated to the dual lattice $\Gamma^*$. Instead, we glue opposite faces of the fundamental cell together, because then we can identify the quotient $\BZ$ with the torus 
\begin{align*}
	\T^d = \Bigl \{ z = (z_1 , \ldots , z_d) \in \C^d \; \big \vert \; \abs{z_j} = 1 , \; j = 1 , \ldots , d \Bigr \} 
\end{align*}
by means of so-called \emph{Floquet multipliers} 
\begin{align*}
	z = \e^{\ii k} := \bigl ( \e^{\ii k \cdot e_1} , \ldots , \e^{\ii k \cdot e_d} \bigr ) 
	= \bigl ( \e^{\ii 2 \pi \, k_1} , \ldots , \e^{\ii 2 \pi \, k_d} \bigr ) 
\end{align*}
This identification is also used when introducing the Fourier transform 
\begin{align}
	(\Fourier c)(k) = \sum_{\gamma \in \Gamma} \e^{- \ii k \cdot \gamma} \, c(\gamma) 
	, 
	&&
	c \in \ell^2(\Gamma) \cap \ell^1(\Gamma) 
	, 
\end{align}
as a map $\Fourier : \ell^2(\Gamma) \longrightarrow L^2(\BZ,\dd k)$ where the measure $\dd k$ is normalized so that $\Fourier$ is unitary. 
We will not overly stress that elements of $\BZ$ are really equivalence classes, and the fact $(\Fourier c)(k - \gamma^*) = (\Fourier c)(k)$ can either be interpreted as $\Gamma^*$-periodicity of the transformed function $\Fourier c$ or equivalently that $\Fourier c$ is well-defined on $\BZ$ since its value is independent of the choice of representative $k\in\BZ$. 

The \emph{Bloch-Floquet transform} $\BF := (\Fourier \otimes \id_{L^2(\WS)}) \circ \Usplit$ is a unitary map 
\begin{align*}
	\BF : L^2(\R^d) \longrightarrow L^2(\BZ) \otimes L^2(\WS) 
\end{align*}
which acts on $\Psi \in L^2(\R^d)$ as 
\begin{align}
	\bigl ( \BF \Psi \bigr )(k) = \sum_{\gamma \in \Gamma} \e^{- \ii k \cdot \gamma} \, T_{-\gamma} \Psi \vert_{\WS}
	\label{bloch_floquet:eqn:gamma_star_periodicity_Bloch_functions0}
\end{align}
and inherits the $\Gamma^*$-periodicity of $\Fourier$ in $k$, 
\begin{align}
	\bigl ( \BF \Psi \bigr )(k - \gamma^*) &= \bigl ( \BF \Psi \bigr )(k) 
	\label{bloch_floquet:eqn:gamma_star_periodicity_Bloch_functions1}
	. 
\end{align}
%

\subsection{Recovering Bloch bands} 
\label{bloch_floquet:bloch_bands}
Now we turn back to the discussion of the operator $H^A$.  Assumption~\ref{bloch_floquet:assumption:HA} and Remark~\ref{rem:vec_pot} allow us to choose a vector potential $A$ representing $B$ whose components are $\Gamma$-periodic with bounded first-order derivatives and we conclude from standard arguments (cf.~\cite[Theorem~XIII.96]{Reed_Simon:M_cap_Phi_4:1978}, for instance) that $H^A$ defines a selfadjoint operator on $H^2(\R^d)$. Note that we do not strive for utmost generality here, presumably, our arguments can be adapted to more general situations. Only lattice periodicity is crucial since it leads to a direct integral decomposition of $H^A$ in crystal momentum $k$, 
\begin{align}
	\BF H^A \BF^{-1} = \int_{\BZ}^{\oplus} \dd k \, H^A(k) := \int_{\BZ}^{\oplus} \dd k \, \Bigl ( \bigl ( -i \nabla_y - A(\hat{y}) \bigr )^2 + V(\hat{y}) \Bigr ) 
	. 
\end{align}
While the operator prescription of $H^A(k)$ is independent of $k \in \BZ$, its domain 
\begin{align*}
	H^2_k(\WS) := \Bigl \{ \varphi \vert_{\WS} \; \big \vert \; &\varphi \in H^2_{\mathrm{loc}}(\R^d) 
	, \;  
	 T_{\gamma} \partial_y^a \varphi = \e^{- \ii k \cdot \gamma} \partial_y^a \varphi 
	\; \forall a \in \N^d , \; \abs{a} \leqslant 1 , \; \forall \gamma \in \Gamma
	\Bigr \}
\end{align*}
is not. More precisely, the Bloch-Floquet transform decomposes the domain of $H^A$, 
\begin{align*}
	\BF : H^2(\R^d) \longrightarrow \int_{\BZ}^{\oplus} \dd k \, H^2_k(\WS)
	, 
\end{align*}
and each $H^A(k)$ is a selfadjoint operator on $H^2_k(\WS)$. It is readily seen that the $k$-dependent Bloch boundary conditions on elements of $H^2_k(\WS)$ and their first-order derivatives are a direct consequence of equation~\eqref{bloch_floquet:eqn:gamma_star_periodicity_Bloch_functions0}. These boundary conditions are well-posed: the Wigner-Seitz cell $\WS$ is a Lipschitz domain and thus, the trace theorem \cite{Lion_Magenes:non_hom_bv_problems:1972} ensures that restricting a $H^2_{\mathrm{loc}}(\R^d)$ function to $\partial \WS$ yields a sufficiently regular function. 

By standard theory \cite{Reed_Simon:M_cap_Phi_4:1978,Berezin_Shubin:Schroeinger_eq:1991}, for each $k \in \BZ$, the operator $H^A(k)$ has purely discrete spectrum accumulating at infinity, $\spec \bigl ( H^A(k) \bigr ) = \{ E_n^B(k) \}_{n \in \N}$. As is customary, we will order the $E_n^B(k)$ in non-decreasing order, \ie we have $E_n^B(k) \leqslant E_{n+1}^B(k)$ for all $n \in \N$, and repeat each according to its multiplicity. The corresponding eigenfunctions $k\mapsto \varphi_n^A(k)$ which satisfy 
\begin{align}
	H^A(k) \varphi_n(k) = E_n^B(k) \, \varphi_n^A(k) 
\end{align}
for each $k$ are called \emph{Bloch functions.} The analyticity of $k \mapsto H^A(k)$ \cite{Kato:perturbation_theory:1995,Reed_Simon:M_cap_Phi_4:1978} implies that the maps $k \mapsto E_n^B(k)$ are continuous everywhere and analytic away from band crossings. The band functions also inherit the $\Gamma^*$-periodicity of $H^A(k) = H^A(k - \gamma^*)$ which follows directly from the definition of the domain $H^2_k(\WS)$. The spectrum of $H^A$ can be expressed in terms of the band functions as 
\begin{align*}
	\spec(H^A) = \bigcup_{n \in \N} E_n^B(\BZ) 
	. 
\end{align*}
Similarly, if we choose the phase of $\varphi_n^A(k)$ for each $k$ properly, we can also ensure that $k \mapsto \varphi_n^A(k)$ are piecewise analytic functions with values in $L^2(\WS)$. However, in general it may not be possible to choose the phases in such a way that Bloch functions are analytic or even just continuous on all of $\BZ$, a fact that is related to the main point of our work. 
\begin{assumption}[Local spectral gap]\label{bloch_floquet:assumption:spectral_gap}
	There exists a family of \emph{relevant bands} $\{ E_n^B \}_{n \in \mathcal{I}}$, with $\mathcal{I} \subset \N$ of finite cardinality $\abs{\mathcal{I}}=m$,  such that 
	\begin{align*}
		\inf_{k \in \BZ} \mathrm{dist} \, \Bigl ( \bigcup_{n \in \mathcal{I}} \{ E_n^B(k) \} , \bigcup_{n \in \N \setminus \mathcal{I}} \{ E_n^B(k) \} \Bigr ) =: C_g > 0
	\end{align*}
	holds. For brevity, we denote the collection of eigenvalues with $\specrel(k) := \bigcup_{n \in \mathcal{I}} \bigl \{ E_n^B(k) \bigr \}$. 
\end{assumption}
The relevant part of the spectrum 
\begin{align*}
	\specrel = \bigcup_{k \in \BZ} \specrel(k) 
\end{align*}
is then recovered as the union of the local spectra. 

Fiber-wise, we can define the projection 
\begin{align}
	P_{\specrel}(k) := 1_{\specrel(k)} \bigl ( H^A(k) \bigr ) 
	= \frac{\ii}{2\pi} \int_{\mathsf{C}(k)} \dd \zeta \, \bigl ( H^A(k) - \zeta \bigr )^{-1} 
	\label{bloch_floquet:eqn:P_k}
\end{align}
onto the relevant eigenvalues $\{ E_n^B(k) \}_{n \in \mathcal{I}}$; it can also be written in terms of a Cauchy integral around the $k$-dependent contour $\mathsf{C}(k)$ which encloses $\specrel(k)$. Alternatively, we can express $P_{\specrel}(k)$ in terms of normalized Bloch functions as 
\begin{align*}
	P_{\specrel}(k) = \sum_{n \in \mathcal{I}} \opro{\varphi_n^A(k)}{\varphi^A_n(k)} 
	. 
\end{align*}
Even though the Bloch functions need not be continuous, the gap condition described in Assumption~\ref{bloch_floquet:assumption:spectral_gap} ensures $P_{\specrel}(k)$ depends analytically on $k$. We can combine all the fiber operators to the projection
\begin{align*}
	P_{\specrel} := \BF^{-1} \, \left(\int_{\BZ}^{\oplus} \dd k \, P_{\specrel}(k)  \right) \, \BF 
\end{align*}
on $L^2(\R^d)$. However, unless there is a global gap, $P_{\specrel}$ cannot be written as $1_{\specrel}(H^A)$ for some set $\specrel$, namely it is not a spectral projection of $H^A$. This means that Assumption~\ref{bloch_floquet:assumption:spectral_gap} is stronger than the existence of a global spectral gap as in equation~\eqref{eq:G1}. Finally, the \emph{geometric rank} of $P_{\specrel}$ defined in Definition 1.1 coincides with $\abs{\mathcal{I}}$. 

\subsection{Gauge-covariance} 
\label{bloch_floquet:gauge_covariance}
We have previously indicated that the family of Bloch functions $\{ E_n^B \}_{n \in \N}$ depends on the magnetic field $B$ rather than a particular choice of vector potential. This is due to covariance, \ie if $A' = A + \nabla \chi$ is an equivalent gauge for $B$, then $H^A$ and
\begin{align*}
       H^{A + \nabla \chi} = \e^{- \ii \chi(\hat{x})} \, H^A \, \e^{+ \ii \chi(\hat{x})}
\end{align*}
are unitarily equivalent and as such isospectral. Hence, the spectrum of $H^A$ depends only on the magnetic field $B$. If $\chi$ is a $\Gamma$-periodic gauge function, then
\begin{align*}
       \BF H^{A+\nabla \chi} \BF^{-1} = \int_{\BZ}^\oplus \dd k \, H^{A + \nabla \chi}(k) = \int_{\BZ}^{\oplus} \dd k \, \e^{-\ii \chi(\hat{y})} H^A(k) \e^{+\ii \chi(\hat{y})}
\end{align*}
also fibers in $k$ and $H^A(k)$ and $H^{A + \nabla \chi}(k)$ are also related via $\e^{- \ii \chi(\hat{y})}$. This implies that also the fiber hamiltonians are unitarily equivalent and thus, the band functions depend only on $B$.

\subsection{Integration and derivation}
\label{bloch_floquet:integration_derivation}
Later on, we will need the notions of trace per unit volume and that of a derivation of bounded operators on $L^2(\R^d)$. These notions have been studied extensively in the last decades in the context of quasi-periodic or random operators. Albeit we are interested only in the periodic case, we refer the reader to \cite{Bellissard:gap_labelling:1993} and references therein for a more detailed discussion. 

Consider an increasing sequence $0\in\Gamma_1\subset\Gamma_2\subset\ldots\subset \Gamma$ of bounded subsets of the lattice $\Gamma$ with the property that $\Gamma_n \nearrow \Gamma$. To this sequence we can associate the sequence $\{\WS_n\}_{n\in\N}$ of subsets of $\R^d$ defined by $\WS_n := \bigcup_{\gamma\in\Gamma_n} \overline{\WS+\gamma}$. A bounded operator $Y\in\bbb{B}(L^2(\R^d))$ admits a \emph{trace per unit volume}  if
\begin{align}
 	\mathcal{T}(Y) := \lim_{n\to \infty} \frac{1}{|\Gamma_n|\ |\WS|}\ \text{Tr}_{L^2(\R^d)} \big(1_{\WS_n}\ Y\ 1_{\WS_n}\big) < \infty
	\label{eq:TPUV}
\end{align}
where $|\Gamma_n|$ denotes the cardinality of $\Gamma_n$, $|\WS|$ is the volume of the Wigner-Sitz cell
and $1_{\WS_n}$ is the characteristic function of $\WS_n$ which acts as a projection $1_{\WS_n}:L^2(\R^d) \hookrightarrow L^2(\WS_n)$. One can show that $\mathcal{T}(A)$ is independent of the particular choice of $\WS_n \nearrow \R^d$, we could have used any other \emph{F{\o}lner sequence} of subsets for $\R^d$. 

The significance of the trace per unit volume to periodic operators is provided by the following 
\begin{lem}\label{lem_int_BF}
	Let $Y$ be a bounded $\Gamma$-periodic operator acting on $L^2(\R^d)$ and 
	\begin{align*}
		\BF \, Y \, \BF^{-1}= \int_{\BZ}^\oplus \dd k \, Y(k) 
	\end{align*}
	be its Bloch-Floquet decomposition. Suppose in addition that $Y(k)$ is trace-class in $L^2(\WS)$ for almost all $k\in\BZ$ and
	\begin{align*}
		k \mapsto \text{Tr}_{L^2(\WS)} \bigl ( Y(k) \bigr ) \in L^1(\BZ)
		. 
	\end{align*}
	Then the trace per unit volume of $Y$ is finite and given by 
	\begin{align*}
		\mathcal{T}(A) = \int_{\BZ} \dd k \, \frac{1}{|\WS|} \, \text{Tr}_{L^2(\WS)}\big ( A(k) \big ) 
		. 
	\end{align*}
\end{lem}
\begin{proof}
	The proof of this result is based on the following two relations which can be easily checked (see \cite[Lemma~3.3]{Panati_Sparber_Teufel:polarization:2006} for more details):
	\begin{align*}
		\text{Tr}_{L^2(\R^d)}\big(1_{\WS} \, Y \, 1_{\WS}\big) = \int_{\BZ} \dd k \, \text{Tr}_{L^2(\WS)} \big( Y(k) \big) < +\infty
	\end{align*}
	and
	\begin{align*}
		\text{Tr}_{L^2(\R^d)}\big(1_{\WS+\gamma} \, Y \, 1_{\WS+\gamma}\big) = \text{Tr}_{L^2(\R^d)}\big(1_{\WS} \, T_{\gamma}^{-1} \, Y \, T_{\gamma} \, 1_{\WS}\big) 
		= \text{Tr}_{L^2(\R^d)}\big(1_{\WS} \, Y \, 1_{\WS}\big) 
	\end{align*}
	holds for all $\gamma\in\Gamma$.
\end{proof}
We denote by $\mathcal{K}_\Gamma^1$ the subset of $\bbb{B} \bigl ( L^2(\R^d) \bigr )$ which consists of $\Gamma$-periodic operators which satisfy the assumptions of Lemma \ref{lem_int_BF}. By means of standard arguments for trace-class operator \cite[Chapter VI]{Reed_Simon:M_cap_Phi_1:1972}, one can prove that $\mathcal{K}^1_{\Gamma}$ is an ideal: if $Y\in \mathcal{K}_\Gamma^1$ and $X$ is bounded and $\Gamma$-periodic then both $XY$ and $YX$ are in $\mathcal{K}_\Gamma^1$. Thus, $\mathcal{T}$ has the trace property, \ie $\mathcal{T}\big([Y,X]\big) = 0$ holds for any $Y \in\mathcal{K}_\Gamma^1$ and $\Gamma$-periodic $X \in \mathcal{B} \bigl ( L^2(\R^d) \bigr )$. 

The second notion that we need is that of \emph{derivatives}: if $Y\in\bbb{B}(L^2(\R^d))$ is a bounded operator on $L^2(\R^d)$ we define its \emph{$j$-th derivative} as
\begin{align}
	\delta_j Y:=-\tfrac{\ii}{ 2\pi} \ [\hat{x}_j , Y]
	,
	&&
	j=1,\ldots,d
	, 
	\label{eq:j-th_der}
\end{align}
where $\hat{x}_j$ denotes the position operator projected in the $j$-th direction of the lattice, namely the multiplication by the function $x\cdot e_j$. It is easy to check that the $\delta_j$ verify all the formal properties of a derivation, \ie they are linear and satisfy the Leibniz rule. We say that $Y$ is of class $\mathcal{C}^s$ if $\delta_{j_1} \circ \cdots \circ \delta_{j_s} (Y) \in \bbb{B}(L^2(\R^d))$ for any choice of the $s$ derivatives. 
\begin{lem}\label{lem_dev_BF}
	Let $Y$ be a bounded operator on $L^2(\R^d)$ which is $\Gamma$-periodic and $\mathcal{C}^1$. Then the following statements hold true:
	\begin{enumerate}[(i)]
		\item $\delta_j Y$ is $\Gamma$-periodic for any $ j=1,\ldots,d$. 
		\item Let $k \mapsto Y(k)$ be the map associated to the Bloch-Floquet decomposition of $Y$. Then the derivative $\partial_{k_j} Y (k)$ is a well-defined bounded operator on $L^2(\WS)$ for almost all $k\in\BZ$ and $ j=1,\ldots,d$. Moreover, we have 
		\begin{align}
			\widetilde{\mathcal{U}}_{\mathrm{BF}} \ \delta_j Y \ \widetilde{\mathcal{U}}_{\mathrm{BF}}^{-1} = \int_{\BZ}^\oplus \dd k \, \partial_{k_j} Y (k)
			\label{eq:deriv1}
		\end{align}
		where $\widetilde{\mathcal{U}}_{\mathrm{BF}} := G \, \BF$ and $G := \int_{\BZ}^\oplus \dd k \, G(k)$ is the unitary operator on $L^2(\BZ) \otimes L^2(\WS)$ defined fiberwise by $G(k)\psi(k,y):=\e^{-\ii k \cdot y} \psi(k,y)$.
	\end{enumerate}
\end{lem} 
\begin{proof} 
	\begin{enumerate}[(i)]
		\item follows simply observing that $T_\gamma\ \hat{x}_j\ T_\gamma^{-1}=\hat{x}_j-(\gamma\cdot e_j)\ \id_{L^2(\R^d)}$. 
		\item A simple computation shows that
		\begin{align*}
			\bigl ( \widetilde{\mathcal{U}}_{\mathrm{BF}} \Psi \bigr )(k,y) = \bigl ( G \, \BF\Psi \bigr )(k,y)=\sum_{\gamma\in\Gamma}{\rm e}^{-\ii k \cdot (y+\gamma)}\ \Psi(y+\gamma)
		\end{align*}
		Let $L(b)$ be the unitary operator on $L^2(\R^d)$ defined by $\big(L(b)\Psi\big) (x):={\rm e}^{-\ii b\cdot x}\ \Psi(x)$ with $b \in \R^d$. After Bloch-Floquet transform, $L(b)$ acts as a translation in crystal momentum, 
		\begin{align*}
			\bigl ( \widetilde{\mathcal{U}}_{\mathrm{BF}} L(b) \Psi \bigr )(k,y) = \sum_{\gamma\in\Gamma}{\rm e}^{-\ii k \cdot (y+\gamma)}\ {\rm e}^{-\ii b\cdot(y+\gamma)}\ \Psi(y+\gamma) 
			= (\widetilde{\mathcal{U}}_{\mathrm{BF}}\Psi)(k+b,y)
			, 
		\end{align*}
		which in particular implies 
		\begin{align}\label{eq:deriv2}
			\widetilde{\mathcal{U}}_{\mathrm{BF}} \, \big( L(b) Y L(b)^{-1} \big)\widetilde{\mathcal{U}}_{\mathrm{BF}}^{-1} 
			= \int_{\BZ}^\oplus\dd k \, Y(k+b)
			.
		\end{align}
		Then from the strong limit
		\begin{align}\label{eq:deriv3}
			\lim_{\varepsilon\to0} \frac{L(\varepsilon \, e_j) \, Y \, L(\varepsilon\, e_j)^{-1} - Y}{\varepsilon} = 2\pi \, \delta_j(Y) 
		\end{align}
		and using the modified Bloch-Floquet transform $\widetilde{\mathcal{U}}_{\mathrm{BF}}$, one deduces the existence of
		\begin{align}\label{eq:deriv4}
			\lim_{\varepsilon\to0} \frac{Y(k_1,\ldots,k_j+\varepsilon2\pi,\ldots,k_d)-Y(k_1,\ldots,k_j,\ldots,k_d)}{\varepsilon}:=2\pi\ (\partial_{ k_j}Y)(k)
		\end{align}
		in the strong sense on $L^2(\WS)$ for almost all $k \in \BZ$. Equation \eqref{eq:deriv1} then follows immediately from \eqref{eq:deriv2} and the definition of the limits \eqref{eq:deriv3} and \eqref{eq:deriv4}.
	\end{enumerate}
\end{proof}
The unitary $\widetilde{\mathcal{U}}_{\mathrm{BF}}$ is also known as the Zak-Bloch-Floquet transform (\cf \cite[Section~3.2]{Panati:triviality_Bloch_bundle:2006} for a comparison to the usual Bloch-Floquet transform). 

A local spectral gap for $\specrel$ (\cf Assumption \ref{bloch_floquet:assumption:spectral_gap}) assures that the map $k\mapsto P_\specrel(k)$ is smooth. The fact that the $\partial_{k_j} P_{\specrel} (k)$ are well-defined bounded operators on $L^2(\WS)$, and the arguments in Lemma \ref{lem_dev_BF} show that the projection $P_\specrel$ is $\mathcal{C}^1$. The link between derivatives and trace per unit volume for projections which share the properties of $P_\specrel$ is established in the next result.%
\begin{prop}\label{prop:trace_deriv}
	Let $P$ be a $\Gamma$-periodic and $\mathcal{C}^1$ orthogonal projection on $L^2(\R^d)$. Assume that after Bloch-Floquet decomposition any fiber projection $P(k)$ has finite constant rank $m$ and fix 
	\begin{align*}
		Q_{ij}(P) := P \, \bigl [ \delta_i P , \delta_j P \bigr ] \, P 
		\quad \text{and} \quad 
		\tilde{Q}_{ij}(P)(k) := P(k) \, \bigl [ \partial_{k_i} P(k) , \partial_{k_j}P (k) \bigr ] \, P(k)
		.
	\end{align*}
	Then $Q_{i_1j_1}(P) \cdots Q_{i_Nj_N}(P) \in \mathcal{K}_\Gamma^1 \,$ for any $i_1,j_1,\ldots,i_N,j_N=1,\ldots,d$ and
	\begin{align*}
		\mathcal{T} \big( Q_{i_1j_1}(P) \cdots Q_{i_Nj_N}(P) \big ) = \frac{1}{|\WS|} \, 
		\int_{\BZ} \dd k \, \text{Tr}_{L^2(\WS)} \big ( \tilde{Q}_{i_1j_1}(P)(k) \cdots \tilde{Q}_{i_Nj_N}(P)(k) \big ) 
		.
	\end{align*}
\end{prop}
\begin{proof} 
	The $\Gamma$-periodicity and the boundness of any factor $Q_{ij}(P)$ follows immediately from the definition. By Lemma~\ref{lem_int_BF}, the product $Q_{i_1j_1}(P) \cdots Q_{i_Nj_N}(P)$ is also in $\mathcal{K}^1_{\Gamma}$. Hence, our operators all have the structure $P Y \negthinspace P$ where $Y$ is bounded and $\Gamma$-periodic. Then the estimate 
	\begin{align*}
		\Babs{\text{Tr}_{L^2(\WS)} \big(P(k) \, Y(k) \, P(k)\big)} \leqslant \sum_{l=1}^m \babs{\bscpro{\psi_l(k)}{Y(k)\psi_l(k)}_{L^2(\WS)}} 
		\leqslant m \bnorm{Y(k)}_{L^2(\WS)}
	\end{align*}
	implies that the function $k \mapsto \abs{\text{Tr}_{L^2(\WS)} \big(P(k) \, Y(k) \, P(k)\big)}$ is bounded by $m \snorm{Y}_{L^2(\R^d)}$ for almost all $k\in\BZ$. Hence, $k \mapsto \text{Tr}_{L^2(\WS)} \big(P(k) \, Y(k) \, P(k) \big)$ is integrable and one obtains $P Y P \in \mathcal{K}^1_{\Gamma}$. 
	
	The final part follows from observing that if $Y\in\mathcal{K}_\Gamma^1$, then 
	\begin{align*}
		\text{Tr}_{L^2(\WS)} \big ( Y(k) \big ) = \text{Tr}_{L^2(\WS)} \big ( G(k)^{-1}Y(k)G(k) \big ) 
	\end{align*}
	holds for almost all $k\in\BZ$. 
\end{proof}
%

\subsection{Wannier functions} 
\label{bloch_floquet:wannier_functions}
A rather simple choice of a Wannier system $\{ w_1^A,\ldots, w_m^A \}$ which spans $\ran P_{\specrel} \subset L^2(\R^d)$ is to set 
$w_n^A(x) :=  \BF^{-1} \varphi_n^A$ (cf.~equation \eqref{intro:eqn:wannier_function}), that is the Wannier functions are the partial Fourier transforms of the Bloch functions $\varphi_n^A$ in $k$. According to the theory of Fourier transforms, there is a direct link between regularity of $k \mapsto \varphi_n^A(k)$ and decay of $w_n^A$: by a variant of the Paley-Wiener theorem \cite[Theorem~2.2~(2)]{Kuchment:exponential_decaying_wannier:2009}, $w_n^A$ decays rapidly if and only if $\varphi_n^A$ is smooth and has exponential decay if and only if $\varphi_n^A$ is analytic in $k$ \cite[Lemma~3.3]{Kuchment:exponential_decaying_wannier:2009}. Since Bloch functions $\varphi_n^A$ are only continuous at  band crossings, the corresponding Wannier functions have polynomial decay and not exponential decay. Hence, we have to generalize our question: is it possible to find \emph{a} family 
\begin{align*}
	\Bigl \{ \psi_j : \BZ \longrightarrow L^2(\WS) \; \big \vert \; j = 1 , \ldots , m \Bigr \}
\end{align*}
such that the $k \mapsto \psi_j(k)$ are globally analytic on $\BZ$ and the set $\{ \psi_1(k) , \ldots , \psi_m(k) \}$ forms an orthonormal basis of $\ran P_{\specrel}(k)$ for all $k$? If such an {analytic} family exists, then the corresponding Wannier system $\{ w_1 , \ldots , w_m \}$, 
\begin{align*}
	w_j(x) := \bigl ( \BF^{-1} \psi_j \bigr )(x) 
	= \int_{\BZ} \dd k \, \e^{- \ii k \cdot (x - [x]_{\WS})} \, \psi_j(k,[x]_{\WS}) 
	, 
\end{align*}
is exponentially localized by the aforementioned Paley-Wiener theorem. 

\subsection{Connecting analyticity to continuity: the Oka principle} 
\label{bloch_floquet:the_oka_principle}
The Oka principle is a ``meta-theorem'' linking complex analysis and homology theory, although its ramifications can be put more simply as \cite[p.~145]{Hoermander:complex_analysis:1990}: \emph{``On a Stein manifold it is `usually' possible to do analytically what one can do continuously.''} In a sense, a Stein manifold $X$ is a complex manifold that supports ``sufficiently many'' holomorphic functions \cite[Definition~5.1.3]{Hoermander:complex_analysis:1990}. 

The Brillouin zone $\BZ \cong \T^d$ is evidently too small, because it is compact and thus cannot support any non-trivial holomorphic functions; thus, in order to use the Oka principle, we have to enlarge the Brillouin zone to $\BZ_a := {\R^d_a}^* / \Gamma^*$ where for $a > 0$, we define 
\begin{align*}
	{\R_a^d}^\ast := \Bigl \{ k \in \C^d \; \big \vert \; \abs{\Im k \cdot e_j} < a , \; j = 1 , \ldots , d \Bigr \} 
	. 
\end{align*}
Using Floquet multipliers, we see that $\BZ_a$ is isomorphic to 
\begin{align*}
	\T_a^d := \Bigl \{ z = (z_1 , \ldots , z_d) \in \C^d \; \big \vert \; \e^{-a} < \abs{z_j} < \e^{+a} , \; j = 1 , \ldots , d \Bigr \} 
	. 
\end{align*}
To see that $\BZ_a$ is a Stein manifold, let us remark that the map $\BZ_a \ni k \mapsto \e^{\ii k} \in \T^d_a$ is holomorphic and one-to-one. Thus, if $\T^d_a$ is a Stein manifold, then so is $\BZ_a$. Since finite carteisan products of Stein manifolds are Stein manifolds \cite[Theorem~1 e), p. 125]{Grauert_Remmert:Stein_spaces:2004}, it suffices to note that the annulus $\T^1_a = \bigl \{ z \in \C \; \vert \; e^{-a} < \abs{z} < e^{+ a} \bigr \}$ is a non-compact Riemannian surface and thus a Stein manifold \cite[p.~134]{Grauert_Remmert:Stein_spaces:2004}. 

Now for $a > 0$ small enough, the projection $P_{\specrel}(k)$ and the relevant band energy functions $E_n^B(k)$ extend analytically from $\BZ$ to $\BZ_a$. Hence, if we can find a family 
\begin{align}
	\Bigl \{ \psi_j : \BZ_a \longrightarrow L^2(\WS) \; \big \vert \; j = 1 , \ldots , d \Bigr \} 
	\label{bloch_floquet:eqn:analytic_family}
\end{align}
of analytic functions such that $\{ \psi_1(k) , \ldots , \psi_m(k) \}$ is an orthonormal basis of $\ran P_{\specrel}(k)$ for each $k \in \BZ_a$, then $\bigl \{ \psi_1 \vert_{\BZ} , \ldots , \psi_m \vert_{\BZ} \bigr \}$ forms an analytic family of functions whose Wannier functions $w_j := \BF^{-1} \psi_j \vert_{\BZ}$ are exponentially localized. 
Thanks to the Oka principle, we need not prove \emph{analyticity}, it suffices to show the existence of a \emph{continuous} family $\{ \tilde{\psi}_1 , \ldots , \tilde{\psi}_m \}$ on $\BZ_a$. Since for $a > 0$ small enough, we can extend any continuous function $\psi : \BZ \longrightarrow L^2(\WS)$ for which $\psi(k) \in \ran P_{\specrel}(k)$ to continuous functions on the enlarged Brillouin zone $\BZ_a$, it suffices to show the existence of a continuous family on $\BZ$ instead of $\BZ_a$. We can summarize the discussion in the following 
\begin{prop}\label{bloch_floquet:prop:continuous_family_bloch_exp_loc_wannier}
	Let Assumptions~\ref{bloch_floquet:assumption:HA} and \ref{bloch_floquet:assumption:spectral_gap} be satisfied. Then there exists an exponentially localized Wannier system associated to $P_{\specrel}$ if and only if there exists a family 
	\begin{align*}
		\Bigl \{ \psi_j : \BZ \longrightarrow L^2(\WS) \; \big \vert \; j = 1 , \ldots , m \Bigr \}
	\end{align*}
	of \emph{continuous} functions forming a rank $m$ orthonormal system of $\ran P_{\specrel}(k) \subset L^2(\WS)$ for any $k \in \BZ$. 
\end{prop}
%


\section{Magnetic symmetries}\label{sec_magn_sym} 
In this section, we show how to associate to a large class of symmetries of 
\begin{align*}
	H^0 = - \Delta + V = (-\ii \nabla_x)^2 + V(\hat{x})
\end{align*}
\emph{``magnetic''} symmetries of $H^A = {(-\ii\nabla_x^A)}^2 + V(\hat{x})$ where for the sake of brevity, we have introduced the covariant derivative $-\ii \nabla_x^A := - \ii \nabla_x - A(\hat{x})$. 
Throughout this section, we assume that $A$ is ``sufficiently regular'', meaning circulations $\int_{[x,y]} A$ along line segments $[x,y] \subset \R^d$ are well defined. 
Unless specifically stated otherwise, we do \emph{not} assume that the magnetic field $B$ or the vector potential $A$ are necessarily $\Gamma$-periodic.

\subsection{Definitions and properties}
\label{sec_magn_sym:general_theory}
To define our class of symmetries, set $\mathcal{M}$ to be the {abelian} algebra of Borel measurable functions on $\R^d$. Elements $V \in \mathcal{M}$ define multiplication operators on $L^2(\R^d)$; since elements of $\mathcal{M}$ need not be essentially bounded, the corresponding multiplication operators need not be bounded, but may define unbounded operators on some suitable domain. Now we specify the class of symmetries of interest: 
\begin{defn}[S-transform]\label{mag_sym:defn:S-transform}
	Let $R \in O(\R^d)$ be an orthogonal matrix. A \emph{S-transform} of type $R$ is a unitary or anti-unitary operator $U_R : L^2(\R^d) \longrightarrow L^2(\R^d)$ such that
	\begin{enumerate}[(i)]
		\item $U_R \, (-\ii \nabla_x) \, {U_R}^{-1} = R (-\ii \nabla_x)$; 
		\item $U_R \, \mathcal{M} \, {U_R}^{-1} \subseteq \mathcal{M}$, namely the conjugation by $U_R$ preserves the multiplicative character of the elements in $\mathcal{M}$.
	\end{enumerate}
\end{defn}
A wide array of well-known symmetries are S-transforms, \eg rotations, reflections, translations and time-inversion. The two conditions are natural in the discussion of Schrödinger operators: (i) S-transforms preserve the Galilean symmetry of kinetic energy $(-\ii \nabla_x)^2 = - \Delta$, \ie $[U_R , \Delta] = 0$ holds. Item~(ii) ensures that $U_R$ maps a Schr\"{o}dinger-type operator $-\Delta + V(\hat{x})$ onto {another} Schr\"{o}dinger-type operator $-\Delta + V'(\hat{x})$ with $V'(\hat{x}): = U_R V(\hat{x}) {U_R}^{-1}$. If in addition the S-transform commutes with the potential, 
\begin{align*}
	[V,U_R] = 0 
	, 
\end{align*}
then $[H^0,U_R] = 0$ holds as well and we say \emph{$U_R$ is a symmetry of $H^0$.} 

Now let us define magnetic S-transforms: 
\begin{defn}[Magnetic S-transforms]\label{mag_sym:defn:mag_S-transform}
	Let $U_R$ be an S-transform and $A$ a vector potential associated to the magnetic field $B$. Then the magnetic symmetry associated to $U_R$ is given by 
	\begin{align*}
		U_R^A := \e^{- \ii \int_{[0,\hat{x}]} \hat{A}} \, U_R 
	\end{align*}
	where the magnetic phase is the operator of multiplication with the exponential of the line integral of 
	\begin{align*}
		\hat{A} := R^{-1} (U_R A {U_R}^{-1}) - A
	\end{align*}
	along the line segment $[0,x]$. 
\end{defn}
\begin{thm}\label{mag_sym:thm:properties_mag_S-transform}
	Let $U_R$ be an S-transform and let $A$ be a vector potential for the magnetic field $B$. 
	\begin{enumerate}[(i)]
		\item $U_R^A$ is unitary or anti-unitary if $U_R$ is unitary or anti-unitary, respectively. 
		\item $U_R^A (-\ii \nabla_x^A) {U_R^A}^{-1} = R (-\ii \nabla_x^A)$
		\item Assume $V$ and $A$ are such that $H^0$ and $H^A$ define selfadjoint operators. Then for any symmetry $U_R$ of $H^0$, the associated \emph{magnetic} S-transform $U_R^A$ is a symmetry of $H^A$, \ie 
		\begin{align*}
			\bigl [ H^0,U_R \bigr ] = 0 \quad \Rightarrow \quad \bigl [ H^A,U_R^A \bigr ] = 0 
			. 
		\end{align*}
	\end{enumerate}
\end{thm}
\begin{proof}
	\begin{enumerate}[(i)]
		\item This follows directly from the unitarity of $\e^{-\ii \int_{[0,\hat{x}]} \hat{A}}$. 
		\item To be more concise, we define $\lambda := \e^{- \ii \int_{[0,\hat{x}]} \hat{A}}$. Then, by a simple computation, we get $\bigl [ (- \ii \partial_{x_j}) \, , \, {\lambda}^{-1} \bigr ] = \hat{A}_j \, {\lambda}^{-1} = {\lambda}^{-1} \, \hat{A}_j$ and thus 
		\begin{align*}
			U_R^A (-\ii \partial_{x_j}^A) {U_R^A}^{-1} &= \lambda U_R \, (-\ii \partial_{x_j}^A) \, {U_R}^{-1} {\lambda}^{-1} 
			\\
			&
			= \lambda \bigl ( R (- \ii \nabla_x) \bigr )_j {\lambda}^{-1} - U_R A_j(\hat{x}) {U_R}^{-1} 
			\\
			&= \lambda \Bigl ( {\lambda}^{-1} \, \bigl ( R (-\ii \nabla_x) \bigr )_j + {\lambda}^{-1} \bigl ( R \hat{A}(\hat{x}) \bigr )_j 
			\Bigr ) 
			- U_R A_j(\hat{x}) {U_R}^{-1} 
			\\
			&= \bigl ( R (- \ii \nabla_x) \bigr )_j +  \bigl ( R \hat{A}(\hat{x}) \bigr )_j - U_R A_j(\hat{x}) {U_R}^{-1} 
			\\
			&= \bigl ( R (-\ii \nabla_x^A) \bigr )_j 
			. 
		\end{align*}
		\item This is a direct consequence of (ii) and the fact that $\e^{- \ii \int_{[0,\hat{x}]} \hat{A}} , V \in \mathcal{M}$ commute: 
		\begin{align*}
			U_R^A \, H^A \, {U_R^A}^{-1} &= U_R^A \, {(-\ii \nabla_x^A)}^2 \, {U_R^A}^{-1} + U_R^A V(\hat{x}) {U_R^A}^{-1}
			\\
			&
			= \bigl ( R (-\ii \nabla_x^A) \bigr )^2 + U_R V(\hat{x}) {U_R}^{-1} 
			\\
			&= {(-\ii \nabla_x^A)}^2 + V(\hat{x}) 
			= H^A
		\end{align*}
	\end{enumerate}
\end{proof}
\begin{remark}\label{mag_sym:defn:symmetry_B}
	Note that magnetic symmetries are genuinely different from their non-mag\-netic counterpart, unless the magnetic field is \emph{compatible} with the symmetry. Assume the S-transform $U_R$ is a symmetry of $H^0$,
	then we say that $B$ is compatible with the symmetry $U_R$ if and only if there exists a vector potential $A$ such that 
	\begin{align*}
		A(\hat{x}) = R^{-1} \bigl ( U_R A(\hat{x}) {U_R}^{-1} \bigr ) 
	\end{align*}
	is satisfied. In that case, $U_R$ is also a symmetry of $H^A$, namely if the magnetic field $B \neq 0$ is compatible with $U_R$, magnetic and non-magnetic symmetry coincide. Needless to say that this is impossible by design for some symmetries, \eg time-reversal. 
\end{remark}
\begin{remark}
	It is clear that the magnetic phase factor is not unique: we could have added gauge transformations before and after. E.~g.~for any $a , b \in \R^d$, the operator 
	\begin{align*}
		\e^{+ \ii \int_{[a,\hat{x}]} A}\ \e^{- \ii \int_{[b,\hat{x}]} R^{-1} (U_R A(\hat{x}) {U_R}^{-1})} 
	\end{align*}
	would also serve to ``magnetize'' S-transforms. 
\end{remark}
%

\subsection{Applications} 
\label{mag_sym:examples}
Now let us discuss some relevant examples:

\paragraph{Translations} 
\label{mag_sym:translations}
Let $T : \R^d \longrightarrow \mathcal{U} \bigl ( L^2(\R^d) \bigr )$ be the unitary representation of translations on $L^2(\R^d)$ via $(T_y \Psi)(x) := \Psi(x - y)$ for all $x , y \in \R^d$. Then for any fixed $y$, the operator $T_y$ is an S-transform of type $\mathbf{1} \in O(d)$. In particular, if $V$ is $\Gamma$-periodic, then for any $\gamma \in \Gamma$, the unitary $T_{\gamma}$ is a symmetry of $H^0$; this symmetry extends to the \emph{magnetic translation} 
\begin{align}\label{eq:sec3_G0}
	T_{\gamma}^A = \e^{-\ii \int_{[\gamma,\hat{x}+\gamma]} A} \, \e^{+ \ii \int_{[0,\hat{x}]} A} \, T_{\gamma}
\end{align}
which is a symmetry of $H^A$. We note that the map $T^A : \Gamma \longrightarrow \mathcal{U} \bigl ( L^2(\R^d) \bigr )$, $\gamma \mapsto T^A_{\gamma}$ is not a group representation, but a \emph{generalized projective} representation of the abelian group $\Gamma$ in the sense of \cite{Mantoiu_Purice_Richard:twisted_X_products:2004}: the composition of two magnetic translations 
\begin{align}
	T^A_{\gamma_1} T^A_{\gamma_2} &= \e^{ -\ii \Phi^B[\hat{x};\gamma_1,\gamma_2]} \, T^A_{\gamma_1 + \gamma_2}
	\label{magn_symm:eqn:composition_mag_translations}
\end{align}
is again a magnetic translation up to a phase factor that is the exponential of a magnetic flux $\Phi_B$. 
Moreover, for constant magnetic fields \eqref{magn_symm:eqn:composition_mag_translations} reduces to the usual magnetic translations \cite{zak-64}. Evidently, if $A$ can be chosen $\Gamma$-periodic, then magnetic and non-magnetic translations coincide as the magnetic factor reduces to the identity. 

\paragraph{Rotations and reflections} 
\label{mag_sym:O_d}
Any $R \in O(d)$ acts on $L^2(\R^d)$ via $(\mathcal{R} \Psi)(x) := \Psi(R^{-1} x)$. The magnetic rotations are then defined as 
\begin{align*}
	\mathcal{R}^A := \e^{- \ii \int_{[0,\hat{x}]} \hat{A}} \, \mathcal{R} 
\end{align*}
where $\hat{A} = R A(R^{-1} \, \cdot \,) - A$. Magnetic and non-magnetic rotations coincide if and only if the vector potential satisfies 
\begin{align*}
	R^{-1} A( \cdot \,) = A( R^{-1} \, \cdot \, ) 
	. 
\end{align*}
An important special case is that of the \emph{parity operator} $\parity$ where $R = - \mathbf{1} \in O(d)$ and which acts as $(\parity \Psi)(x) := \Psi(-x)$. If we can choose an odd vector potential, $A(- x) = - A(x)$, then magnetic parity coincides with ordinary parity. 

\paragraph{Time-reversal} 
\label{mag_sym:time_reversal}
For a spinless particle, time inversion is implemented  by complex conjugation 
\begin{align*}
	(C \Psi)(x) := \Psi^*(x) 
	. 
\end{align*}
It is an anti-unitary S-transform of type $-\mathbf{1} \in O(d)$. Since $A$ is a real-valued multiplication operator, $C A_j C = A_j$ holds and \emph{magnetic time-reversal symmetry $C^A$} is given by 
\begin{align*}
	C^A := \e^{- \ii \int_{[0,\hat{x}]} (A - (- A))} \, C = \e^{- \ii \, 2 \int_{[0,\hat{x}]} A} \, C 
	. 
\end{align*}
$C$ is clearly a symmetry of $H^0$ and thus by Theorem~\ref{mag_sym:thm:properties_mag_S-transform} \emph{magnetic} time-reversal is a symmetry of $H^A$, \ie $\bigl [ C^A , H^A \bigr ] = 0$. 
Moreover, $C^2 = \id_{L^2(\R^d)}$ and the unitarity of $\e^{- \ii 2 \int_{[0,\hat{x}]} A}$ implies that also ${C^A}^2 = \id_{L^2(\R^d)}$. Obviously, magnetic time-reversal never coincides with non-magnetic time-reversal unless $A = 0$ and thus $B = 0$. 

\subsection{$\Gamma$-periodic magnetic symmetries and the Bloch-Floquet transform} 
\label{mag_symm:consequences_of_magnetic_symmetries}
Now we assume again that the vector potential $A$ is $\Gamma$-periodic. This section deals with magnetic symmetries preserving $\Gamma$-periodicity: these factor through the Bloch-Floquet transform and the existence of such $\Gamma$-periodic symmetries induces non-trivial relations between vector spaces associated to different fibers of the direct integral decomposition. In particular, we explore in detail the case of magnetic time-reversal and magnetic parity.
\begin{prop}\label{mag_symm:prop:consequence_mag_time-reversal}
	Let Assumption~\ref{bloch_floquet:assumption:HA} be satisfied and define magnetic time-reversal symmetry $C^A$ as above. Then, the following statements hold true:
	\begin{enumerate}[(i)]
		\item $[C^A , T_{\gamma}] = 0$ for all $\gamma\in\Gamma$ 
		\item Define the anti-unitary $J^A \psi := \e^{- \ii 2 \int_{[0,\hat{y}]} A} \, \psi^*$ on $L^2(\WS)$. Then ${C^A_{\mathrm{BF}}} := \BF \, C^A \, \BF^{-1}$ is the anti-unitary map which acts on $\psi \in L^2(\BZ) \otimes L^2(\WS)$ as 
		\begin{align*}
			\bigl ( C^A_{\mathrm{BF}} \psi \bigr )(k) = \e^{- \ii 2 \int_{[0,\hat{y}]} A} \, \psi^*(-k)
			= J^A \bigl ( \psi(-k) \bigr )
			. 
		\end{align*}
		\item  The fiber hamiltonians $H^A(k)$ and $H^A(-k)$ are related via $J^A$, \ie 
		\begin{align*}
			H^A(k) \, J^A = J^A \, H^A(-k)
		\end{align*}
		holds for all $k \in \BZ$. Moreover, if the relevant bands are separated by a gap in the sense of Assumption~\ref{bloch_floquet:assumption:spectral_gap}, then the fiber projections $P_{\specrel}(k)$ and $P_{\specrel}(-k)$ are also related by $J^A$, \ie 
		\begin{align}
			P_{\specrel}(k) \, J^A = J^A \, P_{\specrel}(-k)
			.
			\label{mag_symm:eqn:mag_time-reversal_projections}
		\end{align}
		holds for all $k \in \BZ$. 
	\end{enumerate}
\end{prop}
\begin{proof}
	\begin{enumerate}[(i)]
		\item follows trivially from the definition of $C^A$ and observing that both the complex conjugation $C$ and the multiplication by the phase $\e^{- \ii 2 \int_{[0,x]} A}$ commute with lattice translations $T_{\gamma}$ since $A$ is $\Gamma$-periodic by assumption. 
		\item From (i), we conclude the operator $\e^{- \ii 2 \int_{[0,x]} A}$ fibers trivially in $k$, 
		\begin{align*}
			\BF \, \e^{- \ii 2 \int_{[0,\hat{x}]} A} \, \BF^{-1} = \int_{\BZ}^{\oplus} \dd k \, \e^{- \ii 2 \int_{[0,\hat{y}]} A} 
			, 
		\end{align*}
		meaning that the fiber operator is independent of $k$. Now the claim follows from direct computation: for any $\Psi \in L^2(\R^d)$, we get 
		\begin{align*}
			\bigl ( \BF C^A \Psi \bigr )(k) &= \sum_{\gamma \in \Gamma} \e^{-\ii k \cdot \gamma} \, \big(T_\gamma C^A \Psi \big) \big \vert_{\WS} 
			= \e^{- \ii 2 \int_{[0,\hat{y}]} A} \; \sum_{\gamma \in \Gamma}\e^{-\ii k \cdot \gamma} \, \big(T_\gamma  \Psi^{\ast}\big) \big \vert_{\WS} 
			\\
			&= \e^{- \ii 2 \int_{[0,\hat{y}]} A} \; \biggl ( \sum_{\gamma \in \Gamma} \e^{+\ii k \cdot \gamma} \, \bigl ( T_\gamma  \Psi \bigr ) \big \vert_{\WS} \biggr )^{\ast}
			\\
			&
			= \e^{- \ii 2 \int_{[0,\hat{y}]} A} \, \bigl ( \BF \Psi \bigr )^*(- k) 
			= J^A \bigl ( (\BF \Psi)(-k) \bigr )
			. 
		\end{align*}
		Since $C^A_{\mathrm{BF}}$ is a composition of two unitary and one anti-unitary operator, it is again anti-unitary. 
		\item If we define 
		\begin{align*}
			H^A_{\mathrm{BF}} := \int_{\BZ}^{\oplus} \dd k \, H^A(k) 
			, 
		\end{align*}
		then $[H^A_{\mathrm{BF}} , C^A_{\mathrm{BF}}] = 0$ follows immediately from $[H^A,C^A] = 0$ (Theorem~\ref{mag_sym:thm:properties_mag_S-transform}). In order to prove the first part of (iii), we remark that if $\varphi_n^A(k)$ is a Bloch function to $E_n^B(k)$, then $(C^A_{\mathrm{BF}} \varphi_n^A)(k)$ is an eigenfunction to $H^A(k)$ with eigenvalue $E_n^B(-k)$ since 
		\begin{align*}
			H^A(k) \bigl ( C^A_{\mathrm{BF}} \varphi_n^A \bigr )(k) &=  \bigl ( C^A_{\mathrm{BF}} H^A_{\mathrm{BF}} \, \varphi_n^A \bigr ) (k) 
			= \e^{- \ii 2 \int_{[0,\hat{y}]} A} \, \bigl ( H^A_{\mathrm{BF}} \varphi_n^A \bigr )^{\ast}(-k)
			\\
			&
			= \e^{- \ii 2 \int_{[0,\hat{y}]} A} \, \bigl ( H^A(-k) \varphi_n^A(-k) \bigr )^{\ast}
			\\
			&
			= \e^{- \ii 2 \int_{[0,\hat{y}]} A} \, E_n^B(-k) \, \bigl ( \varphi_n^A \bigr )^{\ast}(-k) 
			\\
			&
			= E_n^B(-k) \, \bigl ( C^A_{\mathrm{BF}} \varphi_n^A \bigr )(k) 
			. 
		\end{align*}
		The fiberwise relation $H^A(k) \, J^A = J^A \, H^A(-k)$
		follows easily from the relation between $C^A_{\mathrm{BF}}$ and $J^A$ and the density of the basis $\varphi_n^A(-k)$. Morevover, we have also proven $\specrel(-k) = \bigcup_{n \in \mathcal{I}} \bigl \{ E_n^B(-k) \bigr \} \subseteq \specrel(k)$. Exchaning the roles of $k$ and $-k$ yields that the two sets are in fact equal, $\specrel(-k) = \specrel(k)$. Hence, by functional calculus and ${J^A}^2 = \id_{L^2(\WS)}$, it follows 
		\begin{align*}
			P_{\specrel}(k) &= 1_{\specrel(k)} \bigl ( H^A(k) \bigr ) 
			= J^A \, 1_{\specrel(-k)} \bigl ( H^A(-k) \bigr ) \, J^A
			= J^A \, P_{\specrel}(-k) \, J^A
			. 
		\end{align*}
	\end{enumerate}
\end{proof}
Similarly, one can prove the same result for the parity operator $\parity$. 
\begin{prop}\label{mag_symm:prop:consequence_mag_parity}
	Let Assumptions~\ref{bloch_floquet:assumption:HA} and \ref{bloch_floquet:assumption:spectral_gap} be satisfied and assume $\parity$ is a symmetry of $H^0$. The unitary operator ${\parity^A_{\mathrm{BF}}} := \BF \, \parity^A \, \BF^{-1}$ acts on $L^2(\BZ) \otimes L^2(\WS)$ as 
	\begin{align*}
		\bigl ( \parity^A_{\mathrm{BF}} \psi \bigr )(k) = \Pi^A \bigl ( \psi(-k) \bigr ) 
	\end{align*}
	where $\Pi^A$ is the multiplication operator $\e^{+ \ii \int_{[0,\hat{y}]} (A(\cdot) + A(- \, \cdot))}$. 
	Moreover, the fiber hamiltonian and the spectral projection at conjugate points are related by $\Pi^A$, \ie 
	\begin{align*}
		H^A(k) \, \Pi^A &= \Pi^A \, H^A(-k)
		\\
		P_{\specrel}(k) \, \Pi^A &= \Pi^A \, P_{\specrel}(-k)
	\end{align*}
	hold for all $k \in \BZ$. 
\end{prop}
%

\section{The Bloch bundle} 
\label{bloch_bundle}
By the Oka principle (Theorem~\ref{bloch_floquet:prop:continuous_family_bloch_exp_loc_wannier}), we are left to prove the existence of a family of $m$ \emph{continuous} functions $\psi_j : \BZ \longrightarrow L^2(\WS)$ which for fixed $k$ form an orthonormal basis of $\Hrel(k)$; our tool of choice is the theory of vector bundles. In that language, the central question can be rephrased as ``Is the Bloch bundle trivial?'' The explanation of this phrase will be the core of this section.

\subsection{A preliminary definition} 
\label{bloch_bundle:int_def}
Let $\specrel$ be a part of the spectrum of $H^A$ which satisfies  Assumption~\ref{bloch_floquet:assumption:spectral_gap}. Then the family of projections $P_{\specrel}(k)$ is well defined on $\BZ$ (\ie it is $\Gamma^{\ast}$-periodic) and it induces a $k$-dependent decomposition
\begin{align*}
	L^2(\WS) &\cong \ran P_{\specrel}(k) \oplus \bigl ( \ran P_{\specrel}(k) \bigr )^{\perp}
\end{align*}
into the complex vector space $\Hrel(k) := \ran P_{\specrel}(k)$ and its orthogonal complement. The relevant subspace $\Hrel(k)$ has fixed dimension $m = \abs{\mathcal{I}}$. We can think of these spaces as being glued onto the Brillouin zone $\BZ$ and the results is the 
\begin{defn}[Bloch bundle]\label{bloch_bundle:defn:bloch_bundle}
	Assume $H^A$ satisfies Assumption~\ref{bloch_floquet:assumption:HA} and the projection $P_{\specrel}$ satisfies the Gap Condition~\ref{bloch_floquet:assumption:spectral_gap}. Then the Bloch bundle $\bundle_{\specrel} = \bigl ( \bspace_{\specrel} , \BZ , \pi_{\specrel} \bigr )$ is the triple consisting of the topological space 
	\begin{align*}
		\bspace_{\specrel} := \bigsqcup_{k \in \BZ} \Hrel(k) = 
		\Bigl \{ (k , \varphi) \in \BZ \times L^2(\WS) \; \big \vert \; \varphi \in \Hrel(k) \Bigr \} 
	\end{align*}
	where $\bigsqcup$ denotes the disjoint union, the Brillouin zone $\BZ$ and the projection $\pi_{\specrel} : \bspace_{\specrel} \to \BZ$ defined by 
	\begin{align*}
		 \bspace_{\specrel} \ni (k,\varphi) \overset{\pi_{\specrel}}{\longmapsto} k \in \BZ
		. 
	\end{align*}
	The topology on $\bspace_{\specrel}$ is generated by the basis of neighborhoods given by 
	\begin{align}
		\mathcal{U}_{(k,\varphi)}(O,\eps) := \Bigl \{ (k',\varphi') \; \big \vert \; k' \in O , \; 
		\varphi' \in \Hrel(k') , \; 
		\norm{\varphi - \varphi'}_{L^2(\WS)} < \eps \Bigr \}
		\label{bloch_bundle:eqn:basis_neighborhoods}
	\end{align}
	where $(k , \varphi) \in \bspace_{\specrel}$ and $O \subset \BZ$ is an open neighborhood of $k$. 
\end{defn}
We will show in Section~\ref{bloch_bundle:primer_bundle_theory} that this really defines a vector bundle in the sense of Definition~\ref{defi_herm_vec_bun}. Roughly speaking, it remains to show local triviality and thus, up to now, $\bundle_{\specrel}$ is only a \emph{Hilbert bundle} \cite{Fell_Doran:basic_rep_theory_groups_algebras:1988}. 

Let us defer this technical detail for a moment and turn back to the main topic of this paper: since the topology on the fibers $\pi_{\specrel}^{-1}(\{ k \}) = \Hrel(k)$ is induced by that of $L^2(\WS)$, the maps $\{ \psi_1 , \ldots , \psi_m \}$ in Proposition~\ref{bloch_floquet:prop:continuous_family_bloch_exp_loc_wannier} can be interpreted as \emph{continuous sections} $\psi_j : \BZ \longrightarrow \bspace_{\specrel}$, \ie continuous maps which satisfy $\pi_{\specrel} \circ \psi_j = \id_{\BZ}$. Hence, in the jargon of vector bundles, Proposition~\ref{bloch_floquet:prop:continuous_family_bloch_exp_loc_wannier} can be restated as: the existence of an exponentially localized Wannier system is equivalent to the existence of $m$ non-vanishing continuous sections $\{ \psi_1 , \ldots , \psi_m \}$ such that for fixed $k$, they are an orthonormal basis of $\ran P_{\specrel}(k)$.

\subsection{Basic notions of vector bundle theory} 
\label{bloch_bundle:primer_bundle_theory}
In this section, we provide some basic definitions and fundamental facts from the theory of vector bundles. For more background, we refer to the monographs  \cite{milnor-stasheff-74,luke-mishchenko-98,hatcher-09} and the expert reader may skip ahead to Section~\ref{geom_analysis:symmetries_vector_bundles}. 

The fact that the Bloch bundle $\bundle_{\specrel}$ defines a vector bundle is actually the content of a Lemma \ref{lem:BBisVB} saying that it verifies the conditions enumerated in the following 
\begin{defn}[Hermitean vector bundle]\label{defi_herm_vec_bun}
	An rank $m$ hermitean vector bundle $\bundle$ over $X$ is a triple $\bundle := (\bspace,X,\pi)$ consisting of a continuous map $\pi : \bspace \longrightarrow X$ between the topological spaces $\bspace$ (the \emph{total space}) and $X$ (the \emph{base space}) such that 
	\begin{enumerate}[(i)]
		\item $X$ is a $d$-dimensional ($d < \infty$) CW-complex, 
		\item for all $x \in X$, the preimage $\bspace_x := \pi^{-1}(\{ x \})$ carries the structure of a complex vector space, 
		\item there exists an open cover $\{ O_{\alpha} \}$ of $X$ and a family of homeomorphisms 
		\begin{align*}
			h_{\alpha} : \pi^{-1}(O_{\alpha}) \longrightarrow O_{\alpha} \times \C^m
		\end{align*}
		mapping $\pi^{-1}(\{ x \})$ onto $\{ x \} \times \C^m$ for each $x \in O_{\alpha}$, and 
		\item the \emph{transition functions} $g_{\alpha,\beta}:=h_\alpha\circ h_\beta^{-1}$ defined on the overlaps $O_{\alpha,\beta}:=O_{\alpha}\cap O_{\beta}\neq\emptyset$ are continuous functions $g_{\alpha,\beta}:O_{\alpha,\beta} \longrightarrow U(m)$ where $U(m)$ denotes the group of the unitary matrices on $\C^m$.
	\end{enumerate}
\end{defn}
Compared to the standard definition, we have added point~(i): the restriction to base spaces which are also CW-complexes is necessary to have a rich classification theory of vector bundles.  In particular, up to homotopy equivalence, any compact manifold is a CW-complex \cite[Corollary~A.12]{hatcher-02}. 

The most general definition of rank $m$ complex vector bundle requires only properties (ii) and (iii): the former ensures that the fibers are all isomorphic to each other and endowed with the proper structure while the latter, the existence of a \emph{local trivialization}, tells us the fibers are glued together continuously. Thus, we can locally identify a complex rank $m$ vector bundle with $\pi^{-1}(O) \cong O \times \C^m$ where $O \subseteq X$. 

The last item in the definition, property (iv), fixes the \emph{structure group} of $\bundle$ to be $U(m) \subset \mathrm{GL}(m)$. This is always possible for vector bundle over a CW-complex (or more generally over a paracompact space) and it is equivalent to the existence of a scalar product on the fibers which varies in a continuous fashion \cite[Proposition 1.2]{hatcher-09}. This justifies the adjective hermitean.

One way to analyze the structure of hermitean vector bundles is to study maps between vector bundles which are compatible with the bundle structure: let $\bundle_j = \bigl ( \bspace_j , X , \pi_j \bigr )$, $j = 1 , 2$, be two hermitean vector bundles over the same base space $X$. An \emph{$X$-map} is a continuous function $f : \bspace_1 \longrightarrow \bspace_2$ such that the fiber restriction $f_x := f \vert_{\pi_1^{-1}(\{ x \})}$ defines a linear homomorphism between $\pi_1^{-1}(\{ x \})$ and $\pi_2^{-1}(\{ x \})$, \ie it is fiber preserving. The set of such maps is denoted by $\text{Hom}(\bundle_1,\bundle_2)$ while we use the short-hand $\text{End}(\bundle)$ for $\text{Hom}(\bundle,\bundle)$. If $f\in\text{Hom}(\bundle_1,\bundle_2)$ such that the restriction $f_x$ is an isomorphism for any $x\in X$ (which implies that $\bundle_1$ and $\bundle_2$ have same rank), then $f$ is automatically an homeomorphism between $\bspace_1$ and $\bspace_2$ and so it defines an $X$-isomorphism between $\bundle_1$ and $\bundle_2$ \cite[Lemma 1.1]{hatcher-09}. In this case we write $\bundle_1\simeq\bundle_2$. Since isomorphic vector bundles have the same rank we write $\text{Vec}^m_\C(X)$ for the set of the equivalence classes of isomorphic rank $m$ hermitian vector bundles over $X$. Classification theory of vector bundles concerns itself with the description of $\text{Vec}^m_{\C}(X)$ for different $m$ and $X$.

A particularly important element is the \emph{trivial} vector bundle $\epsilon^m := (X \times \C^m , X , \mathrm{proj}_1)$ of rank $m$ where the total space is just the cartesian product of base space and fiber, and the map $\mathrm{proj}_1 : X \times \C^m \longrightarrow X$ is the canonical projection onto the first argument. Thus, we call a rank $m$ vector bundle $\bundle$ topologically trivial if and only if it is isomorphic to $\epsilon^m$. 

The triviality of a hermitean vector bundle $\bundle$ can be characterized in terms of sections: $\bundle$ is trivial if and only if there exists a family of \emph{continuous} sections $\{ \psi_1 , \ldots , \psi_m \}$ such that for each $x$ the vectors $\{ \psi_1(x)  , \ldots , \psi_m(x) \}$ is an orthonormal basis for the fiber $\pi^{-1}(\{ x \}) = \bspace_x$ \cite{hatcher-09}. 

\subsection{Connecting triviality to localization} 
\label{bloch_bundle:construction}
This criterion for triviality of vector bundles allows us to rephrase Proposition~\ref{bloch_floquet:prop:continuous_family_bloch_exp_loc_wannier} in the following way: 
\begin{prop}\label{bloch_bundle:prop:main_question_bundle_language}
	There exists an exponentially localized Wannier system associated to $P_{\specrel}$ if and only if the associated Bloch bundle $\bundle_{\specrel}$ is topologically trivial. 
\end{prop}
The task of Section~\ref{geo_analysis} is to prove the triviality of the Bloch bundle using characteristic classes and to discuss the limitations of this method. 
\begin{remark}
	In the language of vector bundles, the Oka principle \cite[Satz~I, p.~268]{Grauert:analytische_Faserungen:1958} can be stated as \emph{``Two vector bundles over a Stein manifold with the same structure group and fiber are topologically equivalent if and only if they are analytically equivalent.''} Roughly speaking, analytic bundles are those for which a family of analytic transition functions $g_{\alpha,\beta}$ exists. 
	
	The Grauert theorem can then be used to show that triviality of the enlarged Bloch bundle $\bundle_{\specrel \, a} := \bigl ( \bspace_{\specrel \, a} , \BZ_a , \pi_{\specrel \, a})$ (defined analogously to $\bundle_{\specrel}$ over the enlarged Brillouin zone $\BZ_a$) is equivalent to the existence of \emph{analytic} sections (cf.~Section~\ref{bloch_floquet:the_oka_principle}). Since $\BZ$ is a deformation retract of $\BZ_a$, the Grauert theorem in conjunction with the Paley-Wiener theorem implies Proposition~\ref{bloch_bundle:prop:main_question_bundle_language}. 
\end{remark}
Although this is a very well-known fact, for completeness, we will prove that $\bundle_{\specrel}$ is a hermitean vector bundle in the sense of Definition~\ref{bloch_bundle:defn:bloch_bundle}: 
\begin{lem}\label{lem:BBisVB}
	The Bloch bundle $\bundle_{\specrel} = (\bspace_{\specrel},\BZ,\pi)$ is a hermitean vector bundle of rank $m = \abs{\mathcal{I}}$.
\end{lem}
\begin{proof}
	It is evident by construction that (i) and (ii) of Definition \ref{defi_herm_vec_bun} are satisfied: indeed $\BZ \simeq \T^d$ is a $d$-dimensional CW-complex  \cite[Example~2.3]{hatcher-02}
	and $\pi^{-1}(\{ k \}) = \Hrel(k) \simeq \C^m$. Then, to complete the proof, we need to construct local trivializations $h_{O}$ which satisfy (iii) and (iv). 

	Let $k_0 \in \BZ$. Then by continuity of $k \mapsto P_{\specrel}(k)$, there exists a neighborhood $O \subseteq \BZ$ of $k_0$ such that $\bnorm{P_{\specrel}(k) - P_{\specrel}(k_0)}_{L^2(\WS)} < 1$. In this neighborhood, we can define a local trivialization via the Nagy formula \cite{Kato:perturbation_theory:1995}: the operator 
	\begin{align*}
		U(k) := \Bigl ( 1 - \bigl ( P_{\specrel}(k) - P_{\specrel}(k_0) \bigr )^2 \Bigr )^{-\nicefrac{1}{2}} \, \Bigl ( P_{\specrel}(k) P_{\specrel}(k_0) + \bigl ( 1 - P_{\specrel}(k) \bigr ) \bigl ( 1 - P_{\specrel}(k_0) \bigr ) \Bigr )
	\end{align*}
	is a unitary mapping between $\Hrel(k_0)$ and $\Hrel(k)$, \ie $P_{\specrel}(k) = U(k) P_{\specrel}(k_0) U(k)^{-1}$. Moreover the map $k \mapsto U(k)$ is continuous in $O$. Now, let $\{ \chi_1 , \ldots , \chi_m \}$ be an orthonormal basis for  $\Hrel(k_0)$ and define $\chi_j(k) := U(k) \chi_j$ for any $k \in O$. Evidently $\{ \chi_1(k) , \ldots , \chi_m(k) \}$ defines an  orthonormal basis for  $\Hrel(k)$ and so any point $p \in \pi_{\specrel}^{-1}(O)$ can be written in a unique way as $\bigl ( k, \psi_c(k) \bigr )$ with $k \in O$, $c = (c_1 , \ldots , c_m)\in \C^m$ and $\psi_c(k) := \sum_{j=1}^m c_j \, \chi_j(k) \in \Hrel(k)$. The family of maps $\{ h_O \}$ defined through 
	\begin{align*}
		\pi_{\specrel}^{-1}(O) \ni \bigl ( k, \psi_c(k) \bigr ) \stackrel{h_{O}}{\longmapsto} (k,c) \in O \times \C^m
	\end{align*}
	for different open sets $O \subset \BZ$, are the candidates for local trivializations from which to build the transition functions. From its definition and the continuity of $k \mapsto U(k)$, it is clear that $h_O$ is continuous and one-to-one. From the explicit form of $h_O^{-1} : O \times \C^m \longrightarrow \pi_{\specrel}^{-1}(O)$
	\begin{align*}
		h^{-1} ( k , c ) \mapsto \bigl ( k , \psi_c(k) \bigr ) 
		, 
	\end{align*}
	one concludes that also the inverse is continuous. Hence, $h_O$ is a homeomorphism and transition functions constructed from different $h_{O}$s satisfy (iv). 
\end{proof}

\subsection{Magnetic symmetries and the geometry of the Bloch bundle} 
\label{geom_analysis:symmetries_vector_bundles}
In this section we explore the effect of magnetic time-reversal symmetry and magnetic parity on the global geometry of the Bloch bundle $\bundle_{\specrel}$. In fact, the existence of such a symmetry induces relations between fibers on conjugate points of the Bloch bundle as shown in Proposition~\ref{mag_symm:prop:consequence_mag_time-reversal}. In order to formulate and prove our result, we need two more notions from vector bundle theory.

The first is that of the \emph{pullback} of a vector bundle which will turn out to be a powerful tool in classification theory of vector bundles (\cf Section~\ref{geo_analysis}). Given a continuous map $f : Y \longrightarrow X$ between CW-complexes and a hermitean vector bundle $\bundle=(\bspace,X,\pi)$, one can construct another vector bundle $f^{\ast}(\bundle) = (\bspace',Y,\pi')$: we define a map $\tilde{f} : \bspace' \longrightarrow \bspace$ fiber-wise such that for any $y \in Y$, $\tilde{f}_y : {\pi'}^{-1}(\{ y \}) \longrightarrow \pi^{-1} \bigl ( \{ f(y) \} \bigr )$ is a vector space isomorphism. 
The vector bundle $f^{\ast}(\bundle)$ is called the \emph{pullback of $\bundle$ via $f$} and it is unique up to isomorphism \cite[Proposition 1.5]{hatcher-09}. Obviously, the pullback preserves the rank of the  vector bundle.

The second notion is that of the \emph{conjugate} vector bundle: given a rank $m$ hermitean vector bundle $\bundle=(\bspace,X,\pi)$, we can simply ``forget'' about the complex structure and think of each fiber as a real vector space of dimension $2m$. Thus, we obtain the \emph{underlying real vector bundle} of rank $2m$ denoted by $\R(\bundle)$. Observe that  the real vector bundle $\R(\bundle)$ and the original complex vector bundle $\bundle$ both have the same total space, base space and projection map. The conjugate vector bundle $\bundle^{\ast} = (\bspace^\ast,X,\pi)$ of $\bundle$ is defined in the following way: it has the same underlying real vector bundle, \ie $\R(\bundle^{\ast})=\R(\bundle)$, but with ``opposite'' complex structure in each fiber. This means $\bspace^{\ast}$ and $\bspace$ are identical as topological spaces, but in $\bspace^{\ast}$ the usual scalar multiplication is replaced with the conjugate scalar multiplication $\C \times \bspace^{\ast} \ni (\lambda,p) \mapsto \lambda^{\ast} p \in \bspace^{\ast}$ \cite[Chapter~14]{milnor-stasheff-74}. 

Armed with these definitions, we can ``rephrase'' Proposition~\ref{mag_symm:prop:consequence_mag_time-reversal} as: 
\begin{thm}\label{mag_symm:thm:consequence_mag_symmetry_bundle_geometry}
	Let $\bundle_{\specrel}=(\bspace_{\specrel},X,\pi)$ be the Bloch bundle associated to the relevant part $\specrel$ of the spectrum of $H^A$. Then the following statements hold true: 
	\begin{enumerate}[(i)]
	          \item We have $\bundle^{\ast}_{\specrel} \simeq f^{\ast}(\bundle_{\specrel})$ where  $f : \BZ \to \BZ$ is the continuous function defined by $f : k \mapsto -k$.
	          \item If in addition the potential $V$ is invariant under parity, then  $\bundle^\ast_{\specrel}\simeq \bundle_{\specrel}$.
	\end{enumerate}
\end{thm}
\begin{proof}
	\begin{enumerate}[(i)]
		\item We define the bundle $\bundle'_{\specrel} = \bigl ( \bspace_{\specrel}' , \BZ , \pi_{\specrel}' \bigr )$ with total space 
		\begin{align*}
			\bspace'_{\specrel} := \bigsqcup_{k \in \BZ} \Hrel(-k) = \Bigl \{ (k,\varphi) \in \BZ \times L^2(\WS) \; \big \vert \; k \in \BZ , \; \varphi \in \ran P_{\specrel}(-k) \Bigr \}
			.
		\end{align*}
		whose topology is generated by the family of neighborhoods given in equation~\eqref{bloch_bundle:eqn:basis_neighborhoods}
		as well as the continuous projection $\pi_{\specrel}' : (k,\varphi) \mapsto k$. Then for any $k \in \BZ$ and $\varphi \in \Hrel(-k)$, the map $(-k,\varphi) \mapsto (k,\varphi)$ defines a continuous function $\tilde{f}: \bspace'_{\specrel} \longrightarrow \bspace_{\specrel}$ such that $\pi_{\specrel} \circ \tilde{f} = f \circ \pi_{\specrel}'$. Moreover, $\tilde{f}$ restricts to the identity fiberwise, $\pi'^{-1}(\{ k \}) = \Hrel(-k) = \pi^{-1} \bigl ( \{ f(k) \} \bigr )$. This proves that the triple $\bundle_{\specrel}' = (\bspace'_{\specrel},\BZ,\pi_{\specrel}')$ is the pullback of the Bloch bundle by the function $f$, \ie $f^{\ast}(\bundle_{\specrel}) = \bundle'_{\specrel}$. 
		
		In (iii) of Proposition \ref{mag_symm:prop:consequence_mag_time-reversal}, we have proven $J^A \, P_{\specrel}(k) = P_{\specrel}(-k) \, J^A$ which implies that $J^A$ defines a bijection between $\Hrel(k)$ and $\Hrel(-k)$. Now let $\Hrel^{\ast}(k)$ be the Hilbert space $\Hrel(k)$ but endowed with the conjugate scalar multiplication $\lambda \ast \varphi := \lambda^{\ast} \varphi$. Then with abuse of notation, we can see the map $J^A$ as a \emph{linear} isomorphism $J : \Hrel^{\ast}(k) \longrightarrow \Hrel(-k)$. Observing that the collection of the spaces $\Hrel^{\ast}(k)$ defines the conjugate bundle $\bundle^{\ast}_{\specrel}$, one can define the map $\tilde{J}^A$ between $\bundle^{\ast}_{\specrel}$ and $f^{\ast}(\bundle_{\specrel})$ defined by $\tilde{J}^A : (k,\varphi) \mapsto \bigl ( k, J^A(\varphi) \bigr ) $. Obviously, $\tilde{J}^A$ is continuous and fiber preserving. Moreover, it restricts to the isomorphism $J^A$ to any fiber, hence $\tilde{J}^A$ defines an isomorphism between vector bundles. This proves (i).
		\item If $V$ is invariant under parity, then $H^A$ commutes with magnetic parity $\parity^A$ and from Proposition~\ref{mag_symm:prop:consequence_mag_parity} it follows that $\Pi^A$ defines a \emph{linear} isomorphism between $\Hrel(k)$ and $\Hrel(-k)$. This implies $\bundle_{\specrel} \simeq f^{\ast}(\bundle_{\specrel})$. Then (ii) follows from (i).
	\end{enumerate}
\end{proof}
%

\section{Conditions  for the triviality of the Bloch bundle}
\label{geo_analysis}
In this section, we show how and when characteristic classes can be used to prove the triviality of the Bloch bundle. By Proposition~\ref{bloch_bundle:prop:main_question_bundle_language}, the (topological) triviality of the Bloch bundle is equivalent to the existence of an exponentially localized Wannier system. Whether or not time-reversal symmetry suffices to prove triviality depends crucially on the dimension $d$ of the Brillouin zone $\BZ$ and the rank $m$ of the Bloch bundle $\bundle_{\specrel}$ (which coincides with the geometric rank of the Wannier system). 

The cases $d = 1$ is special and will be discussed first. Indeed if $d = 1$, the Brillouin zone $\mathcal{B}^1$ is topologically equivalent to the unit circle $\T^1 = S^1$, and a classical result from classification theory of vector bundles states that independently of $m$, vector bundles over $\T^1$ are automatically trivial, \cf equation~\eqref{geom_analysis:eqn:homotopy_classes_T1} and related comments. This is the abstract explanation for why Kohn's result \cite{Kohn:analytic_properties_Bloch_Wannier:1959} works, but cannot be generalized to higher dimension. 

The single band case, $m = 1$, is also special: here, (magnetic) time-reversal symmetry ensures the triviality in any dimension $d$. This can be traced back to the fact that line bundles are uniquely characterized by a single topological invariant, namely the first Chern class. By Theorem~\ref{geo_analysis:thm:triviality_odd_c_j}, the first Chern class vanishes in the presence of time-reversal symmetry. To our knowledge, the case $m = 1$ is the only one where a constructive proof can be used to show triviality of $\bundle_{\specrel}$ (\cf Section~\ref{geom_analysis:m_equal_1}). 

In all remaining cases, one has to check whether the \emph{stable rank condition} is met: if $d \leqslant 2m$, the vanishing of all Chern classes implies the triviality of the vector bundle by Peterson's classification theorem. Explicit counter examples (\eg the one presented in Section~\ref{geom_analysis:limits_characteristic_classes}) show that $d \leqslant 2m$ is not a technical, but an \emph{essential} condition. Indeed, vector bundles with unstable rank ($d > 2m$) can be non-trivial even if all Chern classes vanish. This means studying characteristic classes of vector bundles with unstable rank is insufficient to ensure triviality, and other methods need to be employed.

\subsection{The case $m = 1$: a constructive proof} 
\label{geom_analysis:m_equal_1}
We start analyzing Bloch bundles $\bundle_{\specrel}$ of rank $1$ which means that the relevant part of the spectrum $\specrel$ is covered by a single isolated energy band. In this case the triviality of  $\bundle_{\specrel}$ can be proven constructively by adapting the arguments of Helffer and Sjöstrand \cite{Helffer_Sjoestrand:mag_Schroedinger_equation:1989} to our case where a zero flux magnetic fields is present. 

\begin{prop}\label{geom_analysis:prop:triviality_m_1}
    Assume $P_{\specrel}$ has geometric rank $1$. Then for any $d$, the Bloch bundle $\bundle_{\specrel}$ is trivial. 
\end{prop}
\begin{proof}
	Due to Proposition~\ref{bloch_floquet:prop:continuous_family_bloch_exp_loc_wannier} it suffices to show the existence of a \emph{continuous} $L^2(\WS)$-valued function $\psi : \BZ \longrightarrow \bspace_{\specrel}$ such that $\psi(k) \in \ran P_{\specrel}(k)$ and $\norm{\psi(k)}_{L^2(\WS)} = 1$ hold for all $k \in \BZ$. We will construct such a function inductively over the dimension $d$. For convenience, we  identify continuous functions $f$ on $\BZ$ with functions on $[-\nicefrac{1}{2},+\nicefrac{1}{2}]^d$ which satifsy $f \vert_{k_j = +\nicefrac{1}{2}} = f \vert_{k_j = -\nicefrac{1}{2}}$ for any $j=1,\ldots,d$. 
   
	   Since $[0,+\nicefrac{1}{2}]^d \subset \BZ$ is contractible, we can find a local section $\psi^{(0)} : [0,+\nicefrac{1}{2}]^d \to \bspace_{\specrel}$.  Moreover, we can assume without loss of generality that $\e^{\ii  \int_{[0,\hat{y}]} A} \psi^{(0)}(0)$ is real, namely
	   \begin{align*}
			\psi^{(0)}(0) =  J^A  \bigl ( \psi^{(0)}(0) \bigr ) .
	   \end{align*}
	   Then, the function 
	   \begin{align*}
	       \phi^{(1)}(k) := 
	       \begin{cases}
	           \psi^{(0)}(k) & k \in [0,+\nicefrac{1}{2}]^d \\
	           J^A\bigl (  \psi^{(0)}(-k) \bigr ) & k \in [-\nicefrac{1}{2},0) \times \{ 0 \}^{d-1} \\
	       \end{cases}
	   \end{align*}
	   extends $\psi^{(0)}$ continuously to $[ -\nicefrac{1}{2} , +\nicefrac{1}{2} ] \times [ 0 , +\nicefrac{1}{2} ]^{d-1}$. As $P_{\specrel}(-k) J^A = J^A P_{\specrel}(k)$, we have $\phi^{(1)}(k) \in \ran P_{\specrel}(k)$ on the domain of $\phi^{(1)}$. Observing that the $\Gamma^*$-periodicity of $k \mapsto P_{\specrel}(k)$ implies that the subspaces $\ran P_{\specrel}(+\nicefrac{1}{2},0)$ and $\ran P_{\specrel}(-\nicefrac{1}{2},0)$ coincide, the functions $\phi^{(1)}(+\nicefrac{1}{2},0)$ and $\phi^{(1)}(-\nicefrac{1}{2},0)$ belong to the same $1$-dimensional subspace and differ only by a phase, $\phi^{(1)}(+\nicefrac{1}{2},0)= \e^{\ii \theta_1}\ \phi^{(1)}(-\nicefrac{1}{2},0)$. 
	   Hence, the function 
	   \begin{align*}
	       \psi^{(1)}(k) := \e^{- \ii \theta_1 k_1} \, \phi^{(1)}(k) 
	   \end{align*}
	   is evidently continuous on  $k \in [-\nicefrac{1}{2},+\nicefrac{1}{2}] \times [0,+\nicefrac{1}{2}]^{d-1}$ and agrees at $k = (\pm \nicefrac{1}{2} , 0)$. Hence, we have extended the local section $\psi^{(0)} : [0,+\nicefrac{1}{2}]^d \to \bspace_{\specrel}$ to a local section $\psi^{(1)} : \mathcal{B}^1 \times [0,+\nicefrac{1}{2}]^{d-1} \to \bspace_{\specrel}$. From the definition of $\phi^{(0)}$ and a simple computation, it follows that 
	\begin{align*}
		\psi^{(1)}(k_1,0) =  J^A  \bigl (\psi^{(1)}(-k_1,0) \bigr ) 
		.
	\end{align*}
	Now let us assume we have found a continuous map $\psi^{(j)} : \mathcal{B}^j \times [0,+\nicefrac{1}{2}]^{d-j} \to \bspace_{\specrel}$ for $2 \leqslant j \leqslant d-1$ which satisfies 
	\begin{align*}
		\psi^{(j)}(\underline{k},0) =  J^A  \bigl (\psi^{(j)}(-\underline{k},0) \bigr )
	\end{align*}
	for any $\underline{k}:=(k_1,\ldots,k_j)\in\mathcal{B}^j$. As before, we start defining a new function 
	   \begin{align*}
	       \phi^{(j+1)}(k) := 
	       \begin{cases}
	           \psi^{(j)}(k) & k \in  \mathcal{B}^j\times [0,+\nicefrac{1}{2}]^{d-j} \\
	           J^A\bigl (\psi^{(j)}(-k) \bigr ) & k \in \mathcal{B}^j\times [-\nicefrac{1}{2},0) \times \{ 0 \}^{d-j-1} \\
	       \end{cases}
	   \end{align*}
	   which extends  $\psi^{(j)}$ continuously to $ \mathcal{B}^j \times [-\nicefrac{1}{2},+\nicefrac{1}{2}] \times [0,+\nicefrac{1}{2}]^{d-j}$. As before, using $\mathrm{dim} \, \ran P_{\specrel}(k) = 1$, the $\Gamma^*$-periodicity of $P_{\specrel}(k)$ and $J^A P_{\specrel}(k) = P_{\specrel}(-k) J^A$, one deduces 
	   \begin{align*}
	       \phi^{(j+1)}(\underline{k},+\nicefrac{1}{2},0) &= \e^{ \ii \theta_{j+1}(\underline{k})} \, \phi^{(j+1)}(\underline{k} , - \nicefrac{1}{2},0) .
	   \end{align*}
	Moreover, the assumptions on $\psi^{(j)}$ force $\theta_{j+1}:\R^{j}\to \R$  to be continuous and to satisfy $\beta_i(\underline{k}):=\theta_{j+1}(k_1,\ldots,k_i-1,\ldots,k_j)-\theta_{j+1}(k_1,\ldots,k_i,\ldots,k_j)\in2\pi\Z$ for all $i=1,\ldots,j$. A straightforward computation shows
	\begin{align*}
		\e^{- \ii \theta_{j+1}(\underline{k})} \ \phi^{(j+1)}(\underline{k},+\nicefrac{1}{2},0) &= \phi^{(j+1)}(\underline{k} , - \nicefrac{1}{2},0) = J^A\bigl (\phi^{(j+1)}(-\underline{k}, +\nicefrac{1}{2},0) \bigr )
		\\
		&
		= \e^{ -\ii \theta_{j+1}(-\underline{k})} \, J^A\bigl ( \phi^{(j+1)}(-\underline{k} , - \nicefrac{1}{2},0)         \bigr )    \\
		&
		= \e^{ -\ii \theta_{j+1}(-\underline{k})} \,( J^A)^2\bigl ( \psi^{(j)}(\underline{k} ,  \nicefrac{1}{2},0)         \bigr ) \\
		&
		= \e^{- \ii \theta_{j+1}(-\underline{k})} \ \phi^{(j+1)}(\underline{k},+\nicefrac{1}{2},0)
		,
	\end{align*}
	and thus $\beta_0(\underline{k}):=\theta_{j+1}(\underline{k})-\theta_{j+1}(-\underline{k})\in 2\pi\Z$. Finally, the continuity of $\beta_0$ and $\beta(0)=0$ prove that $\e^{-\ii \theta_{j+1}}$ is an even function, \ie $\theta_{j+1}(\underline{k})=\theta_{j+1}(-\underline{k}) + 2 \pi \Z$. Similarly, one finds that $\beta_i$ is continuous and $\beta_i(0,\ldots,k_i=\nicefrac{1}{2},\ldots,0)=0$ which implies that $\theta_{j+1}(k_1,\ldots,k_i-1,\ldots,k_j)=\theta_{j+1}(k_1,\ldots,k_i,\ldots,k_j) + 2 \pi \Z$ for all $i=1,\ldots,j$. In other words, $\e^{-\ii \theta_{j+1}} : \mathcal{B}^j \to \R$ is continuous. 
	
	The function 
	\begin{align*}
		\psi^{(j+1)}(k) := \e^{- \ii \theta_{j+1}(\underline{k}) k_{j+1}} \, \phi^{(j+1)}(k) 
	\end{align*}
	is evidently continuous on  $k=(\underline{k},k_{j+1},k') \in \mathcal{B}^j\times [-\nicefrac{1}{2},+\nicefrac{1}{2}] \times [0,+\nicefrac{1}{2}]^{d-j-1}$ and agrees at $k = (\underline{k},\pm \nicefrac{1}{2} , 0)$. Hence, $\psi^{(j+1)}$ defines a continuous map on $\mathcal{B}^{j+1}\times [0,+\nicefrac{1}{2}]^{d-j-1}$. Moreover, from the definition of $\phi^{(j+1)}$ and the parity of $\theta_{j+1}$, it follows after a simple computation that
	$$
	\psi^{(j+1)}(\underline{k},k_{j+1},0) =  J^A  \bigl (\psi^{(j+1)}(-\underline{k},-k_{j+1},0) \bigr ) .
	$$    
	This concludes the proof since after $d$ inductive steps one has obtained a continuous section $\psi : \BZ \longrightarrow \bspace_{\specrel}$.
\end{proof}
%

\subsection{Chern classes and conditions for triviality} 
\label{geom_analysis:characteristic_classes_and_triviality}
Classifying the set of non-isomorphic rank $m$ hermitian vector bundles over a given CW-complex $X$ is extremely difficult and usually beyond reach. One common strategy to extract partial information on the geometry is to study topological invariants called \emph{characteristic classes} \cite{milnor-stasheff-74}. For hermitean vector bundles, the relevant characteristic classes are called \emph{Chern classes} (see below for details). Unfortunately, Chern classes do not characterize bundles completely, unless $m = 1$. However, we are less ambitious here, all we would like to know is whether Chern classes can be used to determine whether or not a bundle $\bundle$ is trivial, \ie whether $\bundle \simeq \epsilon^m \in \vecm(X)$. A classical result by Peterson \cite{peterson-59} tells us that knowing all Chern classes suffices to distinguish between trivial and non-trivial bundles, provided that they have stable rank: 
\begin{thm}[\cite{peterson-59}]\label{teo:pet-59}
	Let $\bundle=(\bspace,X,\pi)$ be a rank $m$ hermitean vector bundle over the CW-complex $X$ of finite dimension $d$ which satisfies the following conditions: 
	\begin{enumerate}[(i)]
		\item $\bundle$ is stable, namely $d \leqslant 2m$. 
		\item The only torsion in $H^{2j}(X,\Z)$ is relatively prime to $(j-1)!$. 
	\end{enumerate}
	Then the vector bundle $\bundle$ is trivial if and only if $c_j(\bundle) = 0$ for all $j = 1 , \ldots , \lfloor \nicefrac{d}{2} \rfloor$ where $c_j(\bundle) \in H^{2j}(X,\Z)$ denotes the $j$-th Chern class of $\bundle$.
\end{thm}
As we will see below, the vanishing of the Chern classes is a necessary condition for the trivialitiy of a vector bundle, but Peterson's theorem tells in what situations this condition is also sufficient. 

The proof of this theorem relies on a lot of tools from algebraic geometry and is beyond the scope of this work. We refer the interested reader to the original publication \cite{peterson-59}.

\paragraph{Cohomology} 
\label{geo_analysis:cohomology_B}
In order to apply Peterson's theorem to the Bloch bundle, we need to verify condition (ii) in the assumptions, that is we need to know more about the integer cohomology groups associated to the torus. 

We will denote the \emph{$j$-th integer cohomology (abelian) group} associated to a CW-complex $X$ by $H^j(X,\Z)$ (the basic notions are covered in standard textbooks of algebraic topology, \eg \cite{hatcher-02}). It is customary to define $H^0(X,\Z) := \Z^{\oplus g}$ where $g$ is the number of path-connected components in $X$. 
The \emph{integer cohomology ring} is by definition the direct sum 
\begin{align}
	H^{\bullet}(X,\Z) := \bigoplus_{j = 0}^{\infty} H^j(X,\Z) 
\end{align}
endowed with the \emph{cup product} $\smallsmile$. The Brillouin zone $\BZ$ is topologically equivalent to the torus $\T^d$ and the cohomology ring 
\begin{align*}
	H^{\bullet}(\BZ,\Z) \simeq \bigwedge_{\Z}(\alpha_1 , \ldots , \alpha_d)
\end{align*}
has the structure of an exterior algebra over $\Z$ generated by finite products 
\begin{align*}
	\alpha_{j_1} \smallsmile \cdots \smallsmile \alpha_{j_k}
	, 
	&&
	1 \leqslant j_1 < \ldots < j_k \leqslant d
	, 
\end{align*}
of $\alpha_{j_l} \in H^1(\BZ,\Z)$ \cite[Example~3.15]{hatcher-02}. In particular this means that $\alpha_j \smallsmile \alpha_k = - \alpha_k \smallsmile \alpha_j$ for $j \neq k$ and $\alpha_j \smallsmile \alpha_j = 0$. Hence, $H^j(\BZ,\Z) = \{ 0 \}$ follows if $j > d$ and $H^j(\BZ,\Z) = \Z^{\oplus n(d,j)}$ is the direct sum of $n(d,j) = \nicefrac{d!}{j! (d-j)!}$ copies of $\Z$ if $0 \leqslant j \leqslant d$. This implies the non-trivial groups in the integer cohomology ring are all \emph{torsion free}, \ie all elements are of infinite order, and the cohomology of the Brillouin zone satisfies condition (ii) of Peterson's theorem. Hence, we obtain the following Corollary: 
\begin{cor}\label{geo_analysis:cor:triviality_Bloch_bundle_Peterson}
	Assume the Bloch bundle $\bundle_{\specrel}$ verifies the stable rank condition, $d \leqslant 2m$. Then $\bundle_{\specrel}$ is trivial if and only if all Chern classes vanish, \ie $c_j(\bundle_\specrel) = 0$ for any $j=1 , \ldots , \lfloor \nicefrac{d}{2} \rfloor$.
\end{cor}
%

\paragraph*{Chern classes} 
\label{ssub:chern_classes}
We conclude this section with a brief introduction to Chern classes, the interested reader may find more material in \cite{milnor-stasheff-74,luke-mishchenko-98}. The notion of pullback bundle introduced in Section~\ref{bloch_bundle:primer_bundle_theory} will play a crucial rôle: since pullbacks of isomorphic bundles are themselves isomorphic, any $f : X \to Y$ induces a map $f^* : \vecm(Y) \to \vecm(X)$ (the pullback preserves the rank and reverses the order of $X$ and $Y$). In particular, pullbacks of trivial bundles are still trivial. 

Furthermore, the structure of the pullback bundle depends only on the homotopy class of $f$: if $f,g : X \to Y$ are homotopic, then $f^*(\bundle) \simeq g^*(\bundle)$ holds and the two pullback bundles represent the same element in $\vecm(X)$. This gives rise to the possibility to study $\vecm(X)$ in terms of homotopy invariants. 

All rank $m$ vector bundles over CW-complexes can be seen as pullbacks of the so-called rank $m$ \emph{universal vector bundle} $\bundle_{\mathrm{u}}^{m} = \bigl ( \bspace_{\mathrm{u}}^{m},G^{m}_{\C},\pi_{\mathrm{u}}^{m} \bigr )$ \cite[Chapter~1.2]{hatcher-09}. It turns out that any rank $m$ bundle $\bundle$ over a CW-complex $X$ can be written as a pullback of the universal vector bundle of rank $m$ with respect to a \emph{classifying map} $f : X \to G^{m}_{\C}$, that is $f^*(\bundle_{\mathrm{u}}^{m}) \simeq \bundle$. As mentioned before, pullbacks of homotopic maps yield isomorphic bundles and hence, this construction only depends on the homotopy class $[f] \in \bigl [ X , G^{m}_{\C}\bigr ]$. In fact, there exists an isomorphism between $\bigl [ X , G^{m}_{\C} \bigr ]$ and $\vecm(X)$ via $[f] \mapsto f^* ( \bundle_{\mathrm{u}}^{m} )$ \cite[Theorem~1.16]{hatcher-09} which we can use to classify $\vecm(X)$ by studying $\bigl [ X , G^{m}_{\C} \bigr ]$. 

While studying $\bigl [ X,G^{m}_{\C} \bigr ]$ for general CW-complexes $X$ is out of reach, in case $X = \mathcal{B}^1 \simeq \T^1$, there is only the trivial homotopy class, 
\begin{align}
	\bigl [ \mathcal{B}^1,G^{m}_{\C} \bigr ] = \pi_1 \bigl ( G^{m}_{\C} \bigr ) = \{ 0 \} 
	, 
	\label{geom_analysis:eqn:homotopy_classes_T1}
\end{align}
which implies all complex vector bundles over $\mathcal{B}^1$ are trivial. 

The second ingredient are topological invariants of the base space $G^{m}_{\C}$ of the universal bundle. This space is constructed as a direct limit of a sequence of spaces: let $m \leqslant k$ be two non-negative integers and define the space $G^{(m,k)}_{\C}$ of $m$-dimensional subspaces of $\C^k$ which can be seen as the set of $m$-dimensional hyperplanes passing through the origin of $\C^k$. Any $G^{(m,k)}_{\C}$ can be endowed with the \emph{structure of a finite CW-complex,} making it into a compact manifold of dimension $2 m (k-m)$ called \emph{Grassmann manifold.} The inclusions $\C^k \subset \C^{k+1} \subset \ldots$ yields inclusions $G^{(m,k)}_{\C} \subset G^{(m,k+1)}_{\C} \subset \ldots$, and we can equip 
\begin{align*}
	G^{m}_{\C} := \bigcup_{k = m}^{\infty} G^{(m,k)}_{\C}
\end{align*}
with the direct limit topology. The resulting paracompact space $G^{m}_{\C}$ is then the set of $m$-dimensional subspaces of $\C^{\infty}$. 

The cohomology ring $H^{\bullet}(G^{m}_{\C},\Z) \simeq \Z [\mathfrak{c}_1 , \ldots , \mathfrak{c}_m]$ can be seen as a ring of polynomials with integer coefficients and $m$ generators. These generators $\mathfrak{c}_j \in H^{2j}(G^{m}_{\C},\Z)$ are called \emph{universal Chern classes} and there are no polynomial relationships between them (cf.~\cite[Theorem 3.9]{hatcher-09} or \cite[Theorem 14.5]{milnor-stasheff-74}). 

The Chern classes of a general vector bundle $\bundle$ are constructed as follows: let $f : X \to G^{m}_{\C}$ be any map such that $f^* ( \bundle_{\mathrm{u}}^{m} ) \simeq \bundle$. Then for any $j$, the $f^* : H^j(G^{m}_{\C} , \Z) \to H^j(X,\Z)$ are homomorphisms between the cohomology groups. We now define the \emph{$j$-th Chern class of $\bundle$} as 
\begin{align}
	c_j(\bundle) := f^*(\mathfrak{c}_j) \in H^{2j}(X,\Z) 
	, 
	&&
	j = 1 , 2 , \ldots 
	. 
\end{align}
A closer inspection of the definition of $f^*$ reveals that it only depends on the homotopy class of $f$, and thus isomorphic bundles have the same family of Chern classes. 

Because trivial vector bundles $\bundle \simeq \epsilon^m$ can be seen as pullbacks of constant maps $f$, the associated homomorphism $f^*$ is trivial. Hence, Chern classes of trivial vector bundles are trivial, $c_j(\bundle) = f^*(\mathfrak{c}_j) = 0$. In this sense, Chern classes can be used as a measure of the ``non-triviality'' of a vector bundle. 

\subsection{Vanishing of odd Chern classes} 
\label{geo_analysis:odd_c_j_zero}
From their very definition, Chern classes are \emph{functorial} objects: given a rank $m$ vector bundle $\bundle$ over $X$ and a continuous function $f:Y\to X$, the Chern classes of the two vector bundles $\bundle$ and $f^{\ast}(\bundle)$ are related by 
\begin{align}
	c_j \bigl ( f^{\ast}(\bundle) \bigr ) = f^{\ast} \bigl ( c_j(\bundle) \bigr )
	, 
	&&
	j = 1 , \ldots , m
	,
	\label{geo_analysis:eqn:functioriality_pullback_Chern_classes}
\end{align}
where on the right $f^{\ast}$ denotes the induced homomorphism between cohomology groups $f^{\ast} : H^{2j}(X,\Z) \to H^{2j}(Y,\Z)$. 

Furthermore, there is a relation between the Chern classes of a bundle $\bundle$ and its conjugate $\bundle^\ast$: up to a sign, the two must agree, \ie $c_j(\bundle^{\ast}) = (-1)^j \, c_j(\bundle)$ holds for any $j = 1 , \ldots , m$ \cite[Lemma~14.9]{milnor-stasheff-74}. 

These two facts can be combined for the Bloch bundle $\bundle_{\specrel}$: in view of Theorem~\ref{mag_symm:thm:consequence_mag_symmetry_bundle_geometry} (i), the conjugate bundle $\bundle_{\specrel}^{\ast}$ is the pullback of the Bloch bundle $\bundle_{\specrel}$ with respect to the map $f : \BZ \to \BZ$, $k \mapsto -k$. Hence, its Chern classes satisfy 
\begin{align}
	f^* \bigl ( c_j(\bundle_{\specrel}) \bigr ) = (-1)^j \, c_j(\bundle_{\specrel}) 
	, 
	&&
	j = 1 , \ldots , m
	. 
	\label{geo_analysis:eqn:time_reversed_chern_classes}
\end{align}
As we will see, this will imply $c_j(\bundle_{\specrel}) = 0$ if $j$ is odd. The argument is very similar, but slightly more elaborate than in the case of parity: in the presence of parity symmetry, the strong relation $f^*(\bundle_{\specrel}) \simeq \bundle_{\specrel}$ implies 
\begin{align*}
	c_j(\bundle_{\specrel}) &= c_j \bigl ( f^*(\bundle_{\specrel}) \bigr ) 
	= c_j(\bundle_{\specrel}^*) 
	= (-1)^j \, c_j(\bundle_{\specrel}) 
\end{align*}
for all $j = 1 , \ldots , m$. If $j$ is odd, this implies $2 c_j(\bundle_{\specrel}) = 0$. Since the cohomology groups of the torus are all torsion free, we conclude $c_j(\bundle_{\specrel}) = 0$. 

Even without parity symmetry, (magnetic) time-reversal symmetry of $H^A$ suffices to ensure the triviality of the odd Chern classes. 
\begin{thm}\label{geo_analysis:thm:triviality_odd_c_j}
	The odd Chern classes associated to the Bloch bundle $\bundle_{\specrel}$ are trivial, \ie 
	\begin{align*}
		c_{2j-1}(\bundle_{\specrel}) = 0
	\end{align*}
	holds for all $j = 1 , \ldots , \lfloor \nicefrac{m}{2} \rfloor$
\end{thm}
\begin{proof}
	In spirit, the proof is the same as for parity, we only need to take the explicit form of $f^*$ into account and know more about the relation between homology and cohomology of $\BZ$. 
	
	We will always identify $\BZ$ with $\T^d$. If $d = 1$, the homology group can be computed explicitly \cite[Example~2.2 and Theorem~2A.1]{hatcher-02}: 
	\begin{align*}
		H_j(\mathcal{B}^1,\Z) = 
		\begin{cases}
			\Z & j = 0 , 1 \\
			\{ 0 \} & j > 1 \\
		\end{cases}
	\end{align*}
	Let $v$ (vertex) and $\ell$ (loop) be the generators of $H_0(\mathcal{B}^1,\Z)$ and $H_1(\mathcal{B}^1,\Z)$, respectively. The map $g:\mathcal{B}^1\to \mathcal{B}^1$ defined by $g:k\mapsto -k$ induces an homomorphisms of homological groups $g_\ast:H_j(\mathcal{B}^1,\Z)\to H_j(\mathcal{B}^1,\Z)$ defined on the generators by $g_\ast(v)=v$ and $g_\ast(\ell)=-\ell$. The first equality comes from the (arbitrary) identification of the vertex $v$ with the point $k=[0]\in\mathcal{B}^1$ and the second is a consequence of the fact that the map $g$ reverses the orientation of the fundamental loop $\ell$. To compute the homology for higher dimensional tori, one can use the Künneth formula \cite[Theorem 3B.6]{hatcher-02}: it provides an isomorphism
	\begin{align*}
			H_j(\mathcal{B}^d,\Z) \simeq \bigoplus_{n\in\N} \bigl ( H_n(\mathcal{B}^{d-1},\Z) \otimes_{\Z}  H_{j-n}(\mathcal{B}^1,\Z) \bigr )
			.
	\end{align*}
	between abelian groups and by recursion, one obtains that the homology of $\BZ$ is torsion free. In particular, one computes $H_j(\mathcal{B}^d,\Z) \simeq \Z^{n(d,j)}$ where $n(d,j) = \nicefrac{d!}{j!(d-j)!}$ for all $0 \leqslant j \leqslant d$ and $n(d,j)$ is zero if $j>d$. The $n(d,j)$ generators of $H_j(\mathcal{B}^d,\Z)$ can be obtained by all the possible tensor products $\ell \otimes \cdots \otimes \ell \otimes v \otimes \cdots \otimes v$ of $j$ copies of $\ell$ end $d-j$ copies of $v$. Now let $f:\mathcal{B}^d\to \mathcal{B}^d$ defined as $f := g \times \cdots \times g$. The induced homomorphism $f_{\ast} : H_j(\mathcal{B}^d,\Z) \to H_j(\mathcal{B}^d,\Z)$ acts as $f_\ast := g_\ast \times \cdots \times g_\ast$, and a simple computation using generators shows 
	\begin{align}
		f_*(\alpha) = (-1)^j \, \alpha 
		\label{geo_analysis:eqn:f_star_action_homology}
	\end{align}
	for all $\alpha \in H_j(\mathcal{B}^d,\Z)$. 
	
	The Universal Coefficient Theorem \cite[Theorem~3.2]{hatcher-02} relates homology and cohomology of the Brillouin zone $\BZ$: the cohomology groups 
	\begin{align}
		H^j(\BZ,\mathbb{F}) \simeq \mathrm{Hom}_{\mathbb{F}}\bigl ( H_j(\BZ,\mathbb{F}) , \mathbb{F} \bigr )
		,
		&& 
		\mathbb{F}=\Z,\Q,\R,\C 
		,
		\label{geo_analysis:eqn:homology_cohomology}
	\end{align}
	are isomorphic to the set of the $\mathbb{F}$-module homomorphisms on $H_j(\BZ,\mathbb{F})$ with values in $\mathbb{F}$. Then for $\mathbb{F}=\Z$ and for each $j = 0 , 1 , \ldots$, we can write the action of $b \in H^j(\BZ,\Z)$ on $\alpha \in H_j(\BZ,\Z)$ via the \emph{Kronecker pairing} $\scpro{ \, \cdot \, }{ \, \cdot \, } : H^j(\BZ,\Z) \times H_j(\BZ,\Z) \longrightarrow \Z$ defined by $\scpro{b}{\alpha} := b(\alpha)$. In view of isomorphism~\eqref{geo_analysis:eqn:homology_cohomology}, the  pairing is exact, \ie we have $\scpro{b}{\alpha} = 0$ for all $\alpha \in H_j(\BZ,\Z)$ if and only if $b = 0$. The other important property of the Kronecker pairing we will use is its functoriality: if $h : \BZ \to \BZ$ is any continuous map, then 
	\begin{align*}
		\bscpro{h^*(b)}{\alpha} = \bscpro{b}{h_*(\alpha)} 
	\end{align*}
	holds where $h^*$ and $h_*$ are the induced homomorphisms in cohomology and homology \cite[Proposition~17.4.2]{Dieck:algebraic_topology:2008}. Now if we plug equations~\eqref{geo_analysis:eqn:time_reversed_chern_classes} and \eqref{geo_analysis:eqn:f_star_action_homology} into the Kronecker pairing, we obtain 
	\begin{align*}
		(-1)^j \, \bscpro{c_j(\bundle_{\specrel})}{\alpha} &= \bscpro{f^* \bigl ( c_j(\bundle_{\specrel}) \bigr )}{\alpha}
		= \bscpro{c_j(\bundle_{\specrel})}{f_*(\alpha)}
		= (-1)^{2j} \, \bscpro{c_j(\bundle_{\specrel})}{\alpha}
		\\
		&= \bscpro{c_j(\bundle_{\specrel})}{\alpha}
	\end{align*}
	for any $\alpha \in H_{2j}(\BZ,\Z)$ and $j = 1 , \ldots , m$. Hence, if $j$ is odd, this means $\bscpro{c_j(\bundle_{\specrel})}{\alpha} = 0$ from which we conclude $c_j(\bundle_{\specrel}) = 0$. 
\end{proof}
%

\subsection{The low-dimensional case: $d = 1 , 2 , 3$} 
\label{geo_analysis:low_d}
We can apply Theorem~\ref{geo_analysis:thm:triviality_odd_c_j} to deduce the triviality of the Bloch bundle $\bundle_{\specrel}$ at least in the low-dimensional case. 
\begin{cor}\label{geo_analysis:cor:low_d}
	If $d = 1 , 2 , 3$, the Bloch bundle $\bundle_{\specrel}$ is trivial independently of $m$. 
\end{cor}
\begin{proof}
	If $m = 1$, then the claim follows for any $d > 0$ from Proposition~\ref{geom_analysis:prop:triviality_m_1}. If $d = 1$ and $m$ arbitrary, this is a consequence of equation~\eqref{geom_analysis:eqn:homotopy_classes_T1}. 
	
	If $d = 2 , 3$ and $m \geqslant 2$, the stable rank condition is satisfied and Corollary~\ref{geo_analysis:cor:triviality_Bloch_bundle_Peterson} applies. Due to the low dimensionality of the base space, only the first Chern class can be non-trivial. Time-reversal symmetry forces $c_1(\bundle_{\specrel}) = 0$ (Theorem~\ref{geo_analysis:thm:triviality_odd_c_j}) and we conclude that the Bloch bundle is trivial. 
\end{proof}
%

\subsection{The case $d=4$: condition for the triviality} 
\label{geom_analysis:d=4}
To ensure the triviality of the Bloch bundle $\bundle_{\specrel}$ for $d = 4$ and $m \geqslant 2$, we need to control the second Chern class, \ie we need to check whether $c_2(\bundle_{\specrel}) = 0$. The purpose of this section is to give an equivalent criterion which is, in principle, accessible to numerical computation.

\paragraph{Differential geometry of the Bloch bundle} 
\label{geom_analysis:d=4:diff_geo}
Since $\BZ$ has the structure of a smooth manifold, the Bloch bundle can also be treated from a differential geometric point of view. Here, the geometry of $\bundle_{\specrel}$ can be studied by means of a given \emph{connection} and its related \emph{curvature}. We refer to \cite{Kobayashi_Nomizu:diff_geo:1996,Chern_Chen_Lam:diff_geo:2000} for more details. A connection is a collection of local matrix valued $1$-forms associated to some open cover of $\BZ$ which are glued together by means of transition functions of $\bundle_{\specrel}$. 
Since $\bundle_{\specrel}$ is a hermitean vector bundle, connections take values in $\mathfrak{u}(m)$, the Lie algebra of $U(m)$. 
Of special interest is the \emph{Berry connection} $\mathcal{A} := (\mathcal{A}_{\alpha\beta})$ (we denote the Lie algebra indices with Greek letters), which in local coordinates is given by
\begin{align*}
    \mathcal{A}_{\alpha \beta }(k):=\sum_{j=1}^d\mathcal{A}^{(j)}_{\alpha \beta}(k)\ \dd k_j,\qquad \mathcal{A}^{(j)}_{\alpha \beta}(k):= \bscpro{\psi_{\alpha}(k)}{\partial_{k_j} \psi_{\beta}(k)}_{L^2(\WS)} 
\end{align*}
where $\{\psi_1(k),\ldots,\psi_m(k)\}$ is any locally smooth basis for $\ran P_{\specrel}(k)$, \eg the Bloch functions. The fact that the above expression defines a connection follows from the transformation rule of the local quantities $\mathcal{A}^{(j)}_{\alpha \beta}(k)$ (\cf \cite{Panati:triviality_Bloch_bundle:2006}).

The \emph{Berry curvature} $K:=(K_{\alpha\beta})$ can be derived from the Berry connection $\mathcal{A}$ by means of the structure equation 
\begin{align*}
     K = \dd \mathcal{A} + \mathcal{A} \wedge \mathcal{A} 
    . 
\end{align*}
As it is evident from the above equation, $K$ is a collection of local $\mathfrak{u}(m)$ valued $2$-forms.

 Functions of the Berry curvature provide a perhaps simplified measure of some aspects of the geometry; the \emph{total (differential) Chern class} 
\begin{align}
    \tilde{c}(\bundle_{\specrel}) := \mathrm{det} \, \Bigl ( 1 + \tfrac{\ii}{2\pi} K \Bigr ) =: 1 + \tilde{c}_1(\bundle_{\specrel}) + \ldots + \tilde{c}_m(\bundle_{\specrel})
\end{align}
is the most prominent example. It is defined in terms of the determinant in the Lie algebra indices and yields a sum of even degree elements $\tilde{c}_j(\bundle_{\specrel}) \in H^{2j}_{\mathrm{dR}}(\BZ)$ of the de Rahm cohomology. The first two non-trivial terms can be explicitly computed to be 
\begin{align}
    \tilde{c}_1(\bundle_{\specrel}) &= \frac{\ii}{2\pi} \mathrm{tr} \, K 
    \\
    \tilde{c}_2(\bundle_{\specrel}) &= \frac{1}{8 \pi^2} \Bigl ( \mathrm{tr} \, \bigl ( K \wedge K \bigr ) - \mathrm{tr} \, K \wedge \mathrm{tr} \, K \Bigr )
    = \frac{1}{8 \pi^2} \mathrm{tr} \, \bigl ( K \wedge K \bigr ) + \frac{1}{2} \tilde{c}_1(\bundle_{\specrel}) \wedge \tilde{c}_1(\bundle_{\specrel}) 
    \label{geom_analysis:eqn:tilde_c_2}
\end{align}
where $\mathrm{tr}$ denotes the trace with respect to the Lie algebra indices. Interestingly, for a given Bloch bundle $\bundle_{\specrel}$ the definition of the $\tilde{c}_j(\bundle_{\specrel})$ is independent (in the sense of de Rahm equivalence classes) of the choice of curvature. A short computation using the Berry connection yields that we can express the first differential Chern class in terms of $P_{\specrel}$ and its derivatives: the local expression is given as 
\begin{align}
	\tilde{c}_1(\bundle_{\specrel}) = \frac{\ii}{2\pi} \sum_{1 \leqslant l < j \leqslant d} \mathrm{Tr}_{L^2(\WS)} \bigl ( \tilde{Q}_{lj}(P_{\specrel}) \bigr ) \, \dd k_l \wedge \dd k_j
	,
\end{align}
where
\begin{align*}
	\tilde{Q}_{lj}(P_{\specrel})(k) := P_{\specrel}(k)\, \big[ \partial_{k_l} P_{\specrel}(k) , \partial_{k_j} P_{\specrel}(k) \big ] \, P_{\specrel}(k) 
	. 
\end{align*}
One can find an expression for $\tilde{c}_2(\bundle_{\specrel})$ which has a similar structure. As will be explained below, in the present context $\tilde{c}_1(\bundle_{\specrel}) = 0$ and we will only need to compute the first term in equation~\eqref{geom_analysis:eqn:tilde_c_2}. Then another straightforward computation yields 
\begin{align}\label{eq:difc_2}
    \tilde{c}_2(\bundle_{\specrel}) = \frac{1}{8 \pi^2} \mathrm{tr} \bigl ( K \wedge K \bigr ) = \frac{1}{(2\pi)^2}  \mathrm{Tr}_{L^2(\WS)} \bigl ( \tilde{W}_{\specrel} \bigr ) \, \dd k,
\end{align}
where $\dd k := \dd k_1 \wedge \dd k_2 \wedge \dd k_3 \wedge \dd k_4$ is the normalized volume form on $\BZ$ and locally
\begin{align*}
	\tilde{W}_{\specrel}(k) :=  \tilde{Q}_{12}(P_{\specrel})(k) \, \tilde{Q}_{34}(P_{\specrel})(k) - \tilde{Q}_{13}(P_{\specrel})(k) \, \tilde{Q}_{24}(P_{\specrel})(k) + \tilde{Q}_{14}(P_{\specrel})(k) \, \tilde{Q}_{23}(P_{\specrel})(k).
\end{align*}
%

\paragraph{Link between differential and topological Chern classes} 
\label{geom_analysis:d=4:de_Rahm_integer_cohomology}
Since $\BZ$ is a smooth manifold, the \emph{de Rahm Theorem} \cite[Theorem 4.3]{Chern_Chen_Lam:diff_geo:2000} states that $H_{\mathrm{dR}}^k(\BZ) \simeq H^k(\BZ,\R)$. Furthermore, the canonical injection $\Z \hookrightarrow \R$ induces a homomorphism 
\begin{align*}
    \jmath : H^k(\BZ,\Z) \longrightarrow H^k(\BZ,\R) \simeq H^k_{\mathrm{dR}}(\BZ)
    . 
\end{align*}
As the homology of $\BZ$ is torsion free, we deduce from the Universal Coefficient Theorems for homology and cohomology that $H^k(\BZ,\R) \simeq H^k(\BZ,\Z)\otimes\R$, which implies that $\jmath$ is injective.
The \emph{Chern-Weil Theory} \cite[Appendix C]{milnor-stasheff-74} proves the relevant fact that  $\jmath$ relates topological and differential Chern classes, \ie $\tilde{c}_j(\bundle_{\specrel}) = \jmath \bigl ( c_j(\bundle_{\specrel}) \bigr )$. The injectivity of $\jmath$ assures that $c_j(\bundle_{\specrel}) = 0$ if and only if $\tilde{c}_j(\bundle_{\specrel}) = 0$. In the following we will refer to $\tilde{c}_j(\bundle_{\specrel})$ is the \emph{differential representative} of the $j$-th Chern class.

\paragraph{Criteria for triviality} 
\label{geom_analysis:d=4:cirteria}
Now we come back to the case $d = 4$ and $m \geqslant 2$. Then the differential Chern class $\tilde{c}_1(\bundle_{\specrel}) = \jmath ({c}_1(\bundle_{\specrel})) = 0$ vanishes by Theorem~\ref{geo_analysis:thm:triviality_odd_c_j} which assures ${c}_1(\bundle_{\specrel})=0$. This justifies the first equality in \eqref{eq:difc_2}.
Hence, we have the following criterion for the triviality of the Bloch bundle: 
\begin{thm}\label{theo:CC_d4}
	Assume $d = 4$ and $m \geqslant 2$. Then the Bloch bundle $\bundle_{\specrel}$ is trivial if and only if 
	\begin{align}\label{eq:int_ist_charg}
		\int_{\mathcal{B}^4} \dd k \, \mathrm{Tr}_{L^2(\WS)} \bigl ( \tilde{W}_{\specrel}(k) \bigr ) = 0
		.
	\end{align}
\end{thm}
\begin{proof}
	By Corollary~\ref{geo_analysis:cor:triviality_Bloch_bundle_Peterson} and  Theorem~\ref{geo_analysis:thm:triviality_odd_c_j}  the triviality of the Bloch bundle is equivalent to $c_2(\bundle_{\specrel}) = 0$. In view of the relation between differential and topological Chern classes, $\tilde{c}_2(\bundle_{\specrel}) = 0$ is equivalent to $c_2(\bundle_{\specrel}) = 0$. To complete the proof we need only shows that \eqref{eq:int_ist_charg} is equivalent to the vanishing of $\tilde{c}_2(\bundle_{\specrel})$.
	
	Since \eqref{eq:difc_2} differs from \eqref{eq:int_ist_charg} only by a constant, we obtain 
	\begin{align*}
		0 = \int_{\mathcal{B}^4}\tilde{c}_2(\bundle_{\specrel}) = \bscpro{\tilde{c}_2(\bundle_{\specrel})}{[\mathcal{B}^4]} 
	\end{align*}
	where in the second equality we have rewritten the integration over $\mathcal{B}^4$ as the Kronecker pairing between $H^4(\mathcal{B}^4,\R)\simeq H_{\mathrm{dR}}^4(\BZ)$ and $H_4(\mathcal{B}^4,\R) \simeq H_4(\mathcal{B}^4,\Z)\otimes\R$. The so-called fundamental class $[\mathcal{B}^4]$ denotes the cycle associated to the orientable manifold $\mathcal{B}^4$ which generates $H_4(\mathcal{B}^4,\Z)$.
	The Universal Coefficient Theorem (homology and cohomology of $\mathcal{B}^4$ are torsion free) assure that the pairing is exact (\cf equation \eqref{geo_analysis:eqn:homology_cohomology}), namely $\bscpro{\tilde{c}_2(\bundle_{\specrel})}{[\mathcal{B}^4]}=0$ if and only if $\tilde{c}_2(\bundle_{\specrel})=0$ and this concludes the proof.
\end{proof}
The quantity on the right-hand side of equation~\eqref{eq:int_ist_charg} (even if it is non-zero) is known as the \emph{instanton charge} in physics. In view of equation~\eqref{geom_analysis:eqn:tilde_c_2}, up to a factor it is equal to $\mathrm{ch}_2(\bundle_{\specrel}) = - \tfrac{1}{2} \tilde{c}_1(\bundle_{\specrel})^2 + \tilde{c}_2(\bundle_{\specrel})$ \cite{Atiyah_Hitchen_Singer:self_duality:1978} and $\bundle_{\specrel}$ is trivial if and only if it has zero instanton charge. 

\subsection{Partially localized Wanier systems} 
\label{geom_analysis:partial_loc}
When $d\geqslant 5$, we are generally unable to prove the triviality of the Bloch bundle. Nevertheless, assuming stable rank, we can at least ensure the existence of $l \leqslant m$ continuous linearly independent sections $\{ \psi_1 , \ldots , \psi_l\}$. By usual arguments (the Oka principle and the Paley-Wiener Theorem), this translates immediately into the existence of $l$ linearly independent and exponentially localized Wannier functions which generate a \emph{partially localized Wannier system} for $\ran P_{\specrel}$.

Our analysis is based on two standard results concerning classification theory of vector bundles. The first is a classical result from K-theory:
\begin{lem}[Theorem 1.2, Chapter 9 of \cite{Husemoller:fiber_bundles:1966}]
\label{geom_analysis:cor:one_extra_section_0}
       Let $\bundle = (\bspace , X , \pi)$ be a rank $m$ hermitian vector bundle over a CW-complex of dimension $d$. Assume that $\bundle$ satisfies the stable rank condition, \ie $\sigma = m - \lfloor \nicefrac{d}{2} \rfloor \geqslant 0$. Then
       \begin{align}
         \bundle \simeq \bundle' \oplus \epsilon^{\sigma},
         \label{geom_analysis:eqn:bundle_splitting}
       \end{align}
       namely $\bundle$ splits into the direct sum of a rank $\sigma$ trivial vector bundle $\epsilon^{\sigma}$ and a non-trivial part $\bundle'$.
\end{lem}
In other words, the above result states that there exists an isomorphism between $\mathrm{Vec}^{m}_{\C}(X)$ and $\mathrm{Vec}^{m - \sigma}_{\C}(X)$ provided that $\sigma\geqslant 0$. Moreover, according to the stability property of Chern classes \cite[Lemma~14.3]{milnor-stasheff-74}, $c_j(\bundle) = c_j(\bundle') \in H^{2j}(X,\Z)$.

The second result is a consequence of obstruction theory:
\begin{lem}[Theorem 22 of \cite{luke-mishchenko-98}]\label{geom_analysis:cor:one_extra_section}
       Let $\bundle = (\bspace , X , \pi)$ be a rank $m$ hermitian vector bundle over a CW-complex of dimension $d=2m$. If $c_{m}(\bundle) = 0$ holds in $H^{2m}(X,\Z)$, then there exists a continuous non-vanishing section.
\end{lem}
Combining Lemma~\ref{geom_analysis:cor:one_extra_section_0} and Lemma~\ref{geom_analysis:cor:one_extra_section} we can give a proof of Theorem~\ref{intro:thm:partial_loc_Wannier}:
\begin{proof}[Theorem~\ref{intro:thm:partial_loc_Wannier}]
       By the Oka principle, the existence of $l$ continuous linearly independent sections $\{ \psi_1 , \ldots , \psi_l \}$ already implies the existence of $l$ analytic sections which for fixed $k$ form an orthonormal system. Then, according to the Paley-Wiener Theorem, the preimage of $\{ \psi_1 , \ldots , \psi_l \}$ under $\BF$ is defines $l$ exponentially localized generators of the Wannier system associated to $P_\specrel$.
       \begin{enumerate}[(i)]
               \item If $\sigma = 0$, \ie in the unstable rank regime and in the borderline case $d = 2m$, nothing needs to be proven. Hence, assume $\sigma > 0$. The splitting \eqref{geom_analysis:eqn:bundle_splitting}
                   \begin{align*}
                     \bundle_{\specrel} \simeq \bundle_{\specrel}' \oplus \epsilon^{\sigma}
                   \end{align*}
               assures the existence of $l=\sigma$ continuous and independent sections for the Bloch bundle $\bundle_{\specrel}$.
               \item If $d = 4k+2$, then $\bundle_{\specrel}'$ is a rank $\nicefrac{d}{2} = 2k + 1$ vector bundle. Since the Chern classes of $\bundle_{\specrel}$ and $\bundle_{\specrel}'$ agree and $c_{2k+1}(\bundle_{\specrel}) = 0$ by Theorem~\ref{geo_analysis:thm:triviality_odd_c_j}, we invoke Lemma~\ref{geom_analysis:cor:one_extra_section} to conclude that there exists at least one more continuous section in $\bundle_{\specrel}$, \ie $l=\sigma + 1$ in total.
       \end{enumerate}
\end{proof}
%

\subsection{Failure of the Chern classes approach in the unstable case} 
\label{geom_analysis:limits_characteristic_classes}
The purpose of this section is to construct an explicit example of a vector bundle over $\mathcal{B}^d$ which is non-trivial but whose Chern classes all vanish. In view of Peterson's theorem, $d = 5$ and $m = 2$ is the simplest possible case. 
Such an example proves that in the unstable rank regime, it is not possible to ensure triviality of the Bloch bundle simply by showing the vanishing of all Chern classes. 
The construction of this vector bundle is based on the ideas of Ekendahl \cite{Ekendahl:counterexample_nontrivial_T5_bundle:2010}, but since his original construction makes use of a lot of advanced concepts of algebraic topology, we will sketch a simpler proof here. 

The idea is to define the vector bundle over $\mathcal{B}^5$ as pullback over a non-trivial rank $2$ bundle over $S^5$. Due to the particularly simple structure of the cohomology ring of spheres, the vector bundle has trivial Chern classes by design. 

Up to isomorphism, rank $m$ hermitean vector bundles over $S^d$ are uniquely determined by how the local trivializations over Northern and Southern hemispheres are glued together at the equator by means of a transition function \cite[Proposition~1.11]{hatcher-09}. Topologically, the equator is isomorphic to $S^{d-1}$ and we have to study the homotopy class $[S^{d-1},U(m)]$ which is the homotopy group $\pi_{d-1} \bigl ( U(m) \bigr )$. These homotopy groups have been computed explicitly \cite[Example~4.53]{hatcher-02}, and for $d = 5$ and $m = 2$, we obtain 
\begin{align*}
  \mathrm{Vec}^2_{\C}(S^5) \simeq [S^4,U(2)] = \pi_4 \bigl ( U(2) \bigr ) = \Z_2 
  . 
\end{align*}
Hence, up to isomorphism there are only two rank $2$ bundles over $S^5$, one is trivial, the other one is not; we denote \emph{the} non-trivial vector bundle over $S^5$ by $\eta$. 

Furthermore, all Chern classes of all bundles over $S^5$ are trivial, indeed the Universal Coefficient Theorem for homology \cite[Theorem~3A.3]{hatcher-02} shows that with the exception of $j = 0 , 5$, the cohomology groups with respect to any unital abelian ring $\mathcal{R}$ of $S^5$ are trivial, 
\begin{align*}
  H^j(S^5,\mathcal{R}) = 
  \begin{cases}
    \{ 0 \} & j \in \N \setminus \{ 0 , 5 \} \\
    \mathcal{R} & j = 0 , 5 \\
  \end{cases}
	.
\end{align*}
This implies $H^2(S^5,\mathcal{R})$ and $H^4(S^5,\mathcal{R})$ are trivial and thus, the Chern classes of any vector bundle over $S^5$ vanish. 
\medskip

\noindent
Any continuous map $g : \mathcal{B}^5 \to S^5$ gives rise to a pullback vector bundle $g^*(\eta)$ over  $\mathcal{B}^5$. By functoriality, the Chern classes of the pullback bundle are trivial, 
\begin{align*}
  c_j \bigl ( g^*(\eta) \bigr ) &= g^* \bigl ( c_j(\eta) \bigr ) = g^*(0) = 0 
  , 
  &&
  j = 1 , 2
  . 
\end{align*}
Lastly, we need to show that there exists a $g$ such that the pullback bundle $g^*(\eta)$ is non-trivial. In other words, we need to study the homotopy classes $[\mathcal{B}^5,S^5] \simeq \Z$. It turns out that these homotopy classes are indexed by the so-called \emph{degree} of their elements  and the pullback of a degree $1$ map yields a non-trivial bundle. For an orientable manifold the degree of a map \cite[Chapter 5]{Hirsch:diff_top:1976} is a generalization of the concept of winding number; for our purposes we can simply state that $g$ has degree $1$ if and only if the induced homomorphism $g_{\ast} : H_5(\mathcal{B}^5,\Z) \to H_5(S^5,\Z)$ is an isomorphism. 

The last ingredient needed to prove the non-triviality of $g^\ast(\eta)$ is the notion of \emph{characteristic classes} which generalize the concept of Chern classes \cite[Chapter~20]{Husemoller:fiber_bundles:1966}. They are maps $\gamma_j$ that associate to any rank $m$ bundle $\bundle$ over $X$ an element of $H^j(X,\mathcal{R})$ where $\mathcal{R}$ is an abelian unital ring which are functorial in the following sense: for any $f : X \to Y$ the relation $\gamma_j \circ f^* = f^* \circ \gamma_j$ holds. Here $f^*$ on the left is the pullback and maps $\vecm(Y)$ to $\vecm(X)$ while the $f^*$ on the right is the induced homomorphism between the $j$-th cohomology groups. The set of the $j$-th characteristic classes with coefficients $\mathcal{R}$ is in one-to-one correspondence with elements of $H^j(G^m_{\C},\mathcal{R})$: if $f : X \to G^m_{\C}$ is such that its pullback of the universal vector bundle is isomorphic to $\bundle$, then we obtain  $\gamma_j$ by pulling back elements of $H^j(G^m_{\C},\mathcal{R})$ with respect to $f$ \cite[Chapter~20]{Husemoller:fiber_bundles:1966}. Since trivial bundles can be seen as the pullback with respect to a constant map, the induced homomorphisms are trivial. Thus, a vector bundle is trivial if and only if there are no non-trivial characteristic classes. 

Coming back to the construction, this implies for some ring $\mathcal{R}$ (indeed for $\mathcal{R} = \Z_2$) there exists a non-trivial characteristic class 
\begin{align*}
	\gamma_5 : \mathrm{Vec}^2_{\C}(S^5) \longrightarrow H^5(S^5,\mathcal{R})
\end{align*}
which detects the non-triviality of the bundle $\eta$ by $\gamma_5(\eta) \neq 0$ (\ie the \emph{obstruction}). Using once again that the integer homology groups of $S^5$ and $\mathcal{B}^5$ are torsion free, we obtain from the Universal Coefficient Theorem that 
\begin{align*}
	H^5(X,\mathcal{R}) \simeq \mathrm{Hom} \bigl ( H_5(X,\Z) , \mathcal{R} \bigr ) 
	, 
	&&
	X = S^5 , \mathcal{B}^5 
	. 
\end{align*}
In view of this identification and the fact that $g$ has degree $1$, it is easy to see that the induced homomorphism $g^*$ defines an isomorphism between the two cohomology groups via 
\begin{align*}
	H^5(S^5,\mathcal{R}) \ni \alpha \mapsto g^*(\alpha) := \alpha \circ g_{\ast} \in H^5(\mathcal{B}^5,\mathcal{R})
\end{align*}
and we conclude $g^* \circ \gamma_5 : \mathrm{Vec}^2_{\C}(\mathcal{B}^5) \to H^5(\mathcal{B}^5,\mathcal{R})$ is a non-trivial characteristic class for the vector bundle $g^*(\eta)$. This means the pullback bundle $g^*(\eta)$ with respect to degree $1$ maps is non-trivial. 

\bibliographystyle{alpha}
\bibliography{denittis}

\end{document}